\documentclass[journal,romanappendices]{IEEEtran}

\usepackage{float}
\usepackage{amsmath,amssymb,amsthm}
\newtheorem{theorem}{Theorem}
\usepackage[numbers,sort&compress]{natbib}

\usepackage{hyperref}
\usepackage{graphicx}
\usepackage{booktabs}
\usepackage{flushend}

\usepackage[T1]{fontenc}
\hypersetup{
  colorlinks=true,
  linkcolor=blue,
  anchorcolor=blue,
  citecolor=blue
}
\usepackage{algorithmic}
\usepackage{algorithm}
\usepackage{tikz}
\usetikzlibrary{shapes.geometric,arrows,positioning,fit,backgrounds,calc}

\begin{document}

\title{Matching Particle-Arrival Rates for Low-Count Image Denoising}
\author{Huazhen~Liu,~\IEEEmembership{Student Member,~IEEE,}
        Zipei Yan,~\IEEEmembership{Student Member,~IEEE,}
        and~Jizhou~Li,~\IEEEmembership{Senior Member,~IEEE}%
\thanks{This work is supported by the National Natural Science Foundation of China (No. T2422017), and the Hong Kong Research Grants Council (No. 21204124). 

H. Liu, Z. Yan and J. Li are with the Department
of Electronic Engineering, The Chinese University of Hong Kong, Shatin, Hong Kong, e-mail: jzli@ee.cuhk.edu.hk.  }%
}

\markboth{Submitted to IEEE Transactions on Image Processing, August~2026}%
{}

\maketitle

\begin{abstract}

Low-count imaging arises when limited illumination or dose yields only a few detected particles in microscopy,
such as photons or electrons, causing measurements to be dominated by signal-dependent Poisson shot noise. 
Supervised denoisers typically require paired clean and noisy observations, 
whereas many self-supervised approaches exploit spatial or temporal redundancy without directly estimating the underlying particle-arrival process. 
We introduce PARM, a self-supervised denoising framework based on particle-arrival-rate matching. 
By modeling particle accumulation over exposure as a Poisson counting process, PARM identifies the conditional particle-arrival-rate field with the posterior mean of the latent clean rate. 
Poisson thinning then splits each observation into complementary lower-exposure views that provide self-supervised rate targets for one another. 
In implementation, the network estimates the normalized spatial allocation of the conditional rate, 
while the observed total count determines its global particle mass. 
Experiments on synthetic data, natural-video sequences corrupted with Poisson noise,
and real low-dose electron microscopy demonstrate leading or competitive performance while effectively preserving fine structures. 
PARM therefore provides a statistically grounded approach to low-count image denoising when clean references are unavailable.

\end{abstract}


\IEEEpeerreviewmaketitle

\vspace{-8pt}
\section{Introduction}

\IEEEPARstart{L}{ow}-count imaging arises in many practical domains, such as nighttime surveillance~\cite{shangguan2023day}, autonomous driving~\cite{rapp2020advances}, and both optical~\cite{morris2015imaging} and electron microscopy~\cite{em,shine}.
Under limited illumination, short exposure time, or low radiation dose, sensors capture only a sparse number of discrete particles, such as photons or electrons.
Stochastic variations in these discrete arrivals introduce signal-dependent Poisson shot noise, degrading underlying structures.

Existing approaches fall into two broad categories: model-based and learning-based methods.
Model-based frameworks incorporate Poisson likelihood~\cite{le2007variational,allard2008total}, variance stabilization~\cite{anscombe1948transformation}, or unbiased risk estimation~\cite{purelet}, while supervised learning methods optimize deep neural networks using paired clean and noisy observations~\cite{chen2018learning,li2025ts}.
However, reliable clean references are seldom available for dynamic scenes and scientific imaging modalities.
Self-supervised methods bypass the requirement for clean targets by exploiting independently corrupted observations~\cite{lehtinen2018noise2noise} or internal redundancies~\cite{krull2019noise2void,batson2019noise2self,em,stbn}. They typically treat each measurement as a noisy image acquired at a fixed exposure, rather than as a sample from an exposure-dependent Poisson counting process.

We introduce \textbf{PARM} (\textbf{P}article \textbf{A}rrival-\textbf{R}ate \textbf{M}atching), a self-supervised framework built upon the particle-arrival-rate dynamics that govern Poisson count observations. 
Inspired by flow and generator matching~\cite{lipman2023flow,holderrieth2025generator}, we formulate particle accumulation as an exposure-indexed probability path. By applying the Kolmogorov forward equation~\cite{bartlett1978introduction} to the resulting count dynamics, we derive the conditional particle-arrival-rate field, which equals the posterior mean of the clean particle-arrival rate and is the theoretical target for PARM. In practice, our model learns the normalized spatial allocation of this field.

Since the latent clean particle-arrival rate is unobservable, PARM constructs self-supervised targets through Poisson thinning. Each raw observation is split into complementary lower-exposure views, whose normalized counts provide finite-exposure rate estimates. To further stabilize learning under extremely low-count regimes, we combine bidirectional rate matching with soft cross-view consistency regularization.

Our contributions are summarized as follows:
\vspace{-1pt}
\begin{enumerate}
\item We formulate low-count image denoising through an exposure-dependent Poisson counting process. Using the Kolmogorov forward equation, we derive the conditional particle-arrival-rate field and establish its equivalence to the posterior mean of the clean particle-arrival rate.

\item We establish finite-exposure rate targets from a single noisy observation through Poisson thinning. Under the generalized Kullback--Leibler (GKL) divergence, we prove that optimization utilizing these self-supervised targets shares the same Bayes-optimal predictor with supervised clean-rate regression.

\item We decompose the conditional particle-arrival-rate field into its total conditional rate and normalized spatial allocation. The network learns the allocation, while inference anchors the global particle mass to the observed total count. Bidirectional rate matching and soft cross-view consistency regularization further stabilize learning in ultra-low-count regimes.

\item Experiments on synthetic Poisson data, standard noisy videos, and real low-dose electron microscopy datasets demonstrate that PARM achieves leading or competitive performance, with quantitative gains and structural preservation under extreme noise.
\end{enumerate}

These results establish a physically grounded framework for self-supervised low-count denoising, particularly for low-dose microscopy where clean reference targets are unavailable. The remainder of this paper is organized as follows. Section~\ref{sec:relatedwork} reviews Poisson and low-count denoising and self-supervised approaches. Section~\ref{sec:methods} develops the particle-arrival-rate formulation and its reconstruction strategies. Section~\ref{sec:experiments} evaluates PARM on synthetic and real data, including ablations and runtime comparisons. Section~\ref{sec:conclusion} summarizes the main findings and outlines directions for future work.

\section{Related Work}
\label{sec:relatedwork}
\subsection{Poisson and Low-Count Image Denoising} \label{sec:poisson_image_denoising}
Low-count image observations are typically characterized by the conditional Poisson model \(Y\mid X\sim\operatorname{Poisson}(\Lambda X)\), where \(X\geq0\) is the latent intensity and \(\Lambda>0\) represents the effective exposure or overall count scale.

Model-based methods, such as variational frameworks~\cite{le2007variational}, combine the Poisson likelihood with priors like total variation~\cite{rudin1992nonlinear,allard2008total} to restore corrupted observations.
In particular, PURE-LET~\cite{luisier2010fast} optimizes a linear expansion of transform-domain thresholding functions via a Poisson unbiased risk estimate.
Alternatively, variance-stabilizing transformations, notably the Anscombe transform~\cite{anscombe1948transformation}, are employed to Gaussianize Poisson observations, thereby enabling the deployment of conventional Gaussian denoisers~\cite{blu2007sure,dabov2007image,makinen2020collaborative,maggioni2012nonlocal} for Poisson noise.
To improve the performance in ultra-low-count regimes, exact unbiased inverse transformations and low-count corrections have been developed to reduce reconstruction bias when only a sparse number of particles are detected~\cite{makitalo2010optimal,makitalo2013exact,shen2018improved,zhang2021poisson}.

We note that real imaging sensors are often described by Poisson--Gaussian models, in which the Poisson component captures signal-dependent shot noise and the Gaussian component approximates signal-independent electronic readout noise~\cite{foi2008practical,zhang2019poisson}. In this context, unbiased risk estimation provides a means of selecting restoration parameters directly from noisy measurements without clean ground truth. In particular, PURE-LET was subsequently extended from pure Poisson noise to mixed Poisson--Gaussian noise~\cite{purelet}, and related Monte-Carlo SURE schemes enable black-box parameter optimization under mixed noise~\cite{ramani2008monte,lemontagner2014unbiased}.

More recently, data-driven approaches using deep neural networks have been applied to low-count denoising, ranging from supervised models trained on paired short- and long-exposure observations~\cite{chen2018learning} to diffusion models that exploit learned generative priors~\cite{li2025ts}. However, many supervised frameworks require accurately aligned reference targets, which are often practically infeasible to acquire for dynamic scenes and scientific imaging. Moreover, existing methods typically treat distinct particle levels as isolated degradation states, rather than explicitly modeling how particle counts accumulate over exposure time.

\subsection{Self-Supervised Denoising Approaches}
Self-supervised denoising avoids the requirement for clean targets by exploring self-supervision from noisy observations only. Blind-spot methods, such as Noise2Self~\cite{batson2019noise2self}, Noise2Void~\cite{krull2019noise2void}, PN2V~\cite{krull2020probabilistic}, Noise2Same~\cite{xie2020noise2same}, Self2Self~\cite{quan2020self2self}, and Blind2Unblind~\cite{wang2022blind2unblind}, use masking strategies to enable single-image-based self-denoising.
Complementing these approaches, subsampling-based methods, such as Neighbor2Neighbor~\cite{huang2021neighbor2neighbor}, Noise2Fast~\cite{lequyer2022fast}, ZSN2N~\cite{mansour2023zero}, and Pixel2Pixel~\cite{ma2025pixel2pixel}, construct surrogate training pairs from neighboring pixels.

For dynamic and scientific measurements, representative self-supervised methods include UDVD~\cite{udvd}, UMVD~\cite{aiyetigbo2024umvd}, SHINE~\cite{shine}, and STBN~\cite{stbn}, which exploit temporal redundancy or interpolation to restore noisy video and microscopy data. In related modalities, SSD~\cite{chen2026ssd} uses spectral structure for low-light Raman hyperspectral reconstruction, whereas bit2bit~\cite{bit2bit} models binary particle detections via a Bernoulli observation process. While these methods exploit domain-specific spatial, temporal, or spectral structure, they do not explicitly characterize the exposure-dependent accumulation of Poisson counts.

Several approaches more explicitly address non-Gaussian noise. Self-supervised Poisson--Gaussian denoising extends blind-spot learning to mixed-noise settings~\cite{khademi2021self}; Poisson2Sparse~\cite{ta2022poisson2sparse} combines single-image self-supervision with an unrolled sparse-representation network; and GR2R~\cite{monroy2025generalized} generalizes recorruption-based learning to Poisson and other exponential-family distributions. These methods accommodate Poisson statistics, yet they rely on masking, sparse priors, or synthetic recorruption rather than estimating an exposure-dependent particle-arrival rate.
PARM estimates this rate from a single observation by Poisson thinning, which splits the measurement into complementary lower-exposure views that remain independent given the latent clean rate and therefore serve as self-supervised targets.

\subsection{Discrete Probability Paths and Count-Process Modeling}

Flow matching constructs generative models by learning a continuous velocity field along a prescribed probability path~\cite{lipman2023flow,albergo2025stochastic}. Beyond generative sampling, flow-matching models have also been adapted to image inverse problems.
For instance, Pourya et al.~\cite{pourya2026flower} use a pretrained flow model to address linear inverse problems through flow-consistent destination estimation, measurement-aware refinement, and time progression, interpreting this procedure as an approximate Bayesian posterior sampling.

While these methods operate on continuous image-valued trajectories, recent work has extended this idea to discrete and stochastic state spaces.
Specifically, Discrete Flow Matching~\cite{gat2024discrete} characterizes probability evolution over discrete states, whereas Generator Matching~\cite{holderrieth2025generator} formulates learning through Markov generators or transition rates.
These perspectives motivate transition-rate fields as the discrete counterparts to continuous velocity fields.

Parallel to these developments, count-valued generative models have also incorporated thinning mechanisms and Poisson processes.
For instance, Learning to Jump~\cite{chen2023learning} employs forward thinning and reverse thickening to model latent counts, while ItDPDM~\cite{bhattacharya2025itdpdm} develops a discrete Poisson diffusion framework that completely avoids continuous relaxation.
However, these approaches focus primarily on unconditional sample generation and generative density modeling, do not explicitly use the physical particle-arrival process for conditional image denoising.

While PARM draws conceptual motivation from these frameworks, it differs in objective and formulation.
Unlike related inverse solvers that rely on pretrained continuous generative priors, PARM directly models the physical accumulation of discrete particle counts with respect to exposure.
Its derived transition-rate field provides a principled statistical target for restoration, enabling the network to learn the normalized spatial allocation of particle-arrival rates directly from noisy observations.

\section{Methods}
\label{sec:methods}

\subsection{Overview: From Observed Counts to Particle-Arrival Rates}
\label{sec:overview}

PARM reframes low-count image denoising as conditional particle-arrival-rate estimation.
An observed count image is a stochastic realization of an underlying particle-arrival process.
Therefore, rather than directly regressing a clean image from noisy counts, PARM estimates the rate field that statistically explains the observed count state.

Let $X\in\mathbb{R}_{\geq0}^{n}$ denote the latent intensity and let $Y_\tau$ denote the particle counts acquired at effective exposure $\tau$.
The central statistical target of PARM is the conditional particle-arrival-rate field
\begin{equation}
v_\tau(y)=\mathbb{E}[\lambda_0X\mid Y_\tau=y],
\label{eq:marginal_birth_rate_revised}
\end{equation}
where $\lambda_0>0$ converts latent intensity into particle-arrival rate per unit effective exposure.
Thus, the observed counts constitute the stochastic measurement, whereas $v_\tau(y)$ represents the posterior mean rate underlying that measurement.

Since the clean rate $\lambda_0X$ is usually unavailable, PARM constructs self-supervised targets by complementary Poisson thinning.
A single count observation is split into two lower-exposure views that are conditionally independent given $X$.
Each view is used as the network input, and the complementary counts provide an exposure-normalized rate target for that input.
Under the calibrated Poisson model, matching these targets under the GKL divergence recovers the same population-optimal predictor as supervised clean-rate regression: \(v_{p\tau}\) at a fixed splitting ratio, and the exposure-marginalized predictor \(\mathbb{E}[\lambda_0X\mid Z=z]\) when the ratio is randomized and not supplied to the network.
For practical reconstruction, PARM learns the normalized spatial allocation of the conditional rate, while the observed total count restores its global particle mass, yielding an observed-mass-preserving reconstruction.

\subsection{Exposure-Indexed Poisson Pure-Birth Process}
\label{sec:poisson_probability_path}

Let \(\mathbb{N}_0=\{0,1,2,\ldots\}\). For an image with \(n\) pixels or voxels, let \(X\in\mathbb{R}_{\geq 0}^{n}\) be a nonnegative random vector representing the latent intensity and let \(Y_\tau\in\mathbb{N}_0^{n}\) be a random vector representing the particle-count state at effective exposure \(\tau\).

We first reparameterize the physical acquisition time by a calibrated effective exposure.
Conditioned on a latent clean image \(X=x\), particle arrivals at physical time \(s\) have instantaneous rate \(\lambda(s)x\), where \(\lambda(s)>0\) is a time-varying scalar acquisition factor.
With cumulative exposure \(\mathcal{E}(s)=\int_0^s\lambda(r)\,dr\) and a reference calibration constant \(\lambda_0>0\), the effective exposure \(\tau=\mathcal{E}(s)/\lambda_0\) yields
\begin{equation}
Y_\tau\mid X=x\sim\operatorname{Poisson}(\lambda_0\tau x).
\label{eq:poisson_count_model_revised}
\end{equation}
Here and throughout, the multivariate Poisson distribution is evaluated element-wise, assuming conditionally independent pixel counts given \(X\).
Fig.~\ref{fig:poisson_path} summarizes the exposure-indexed count distributions, conditional particle-arrival-rate fields, and Poisson-thinning construction underlying PARM.

\begin{figure}[t]
    \centering
    \includegraphics[trim={0 50 0 0}, clip, width=\linewidth]{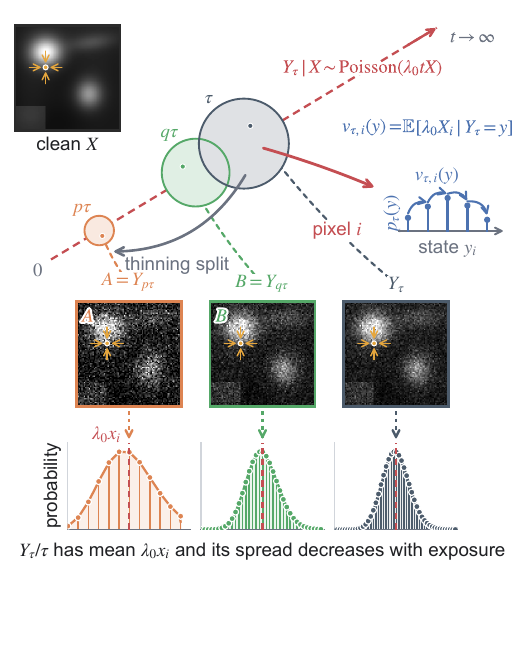}
    \caption{
    Illustration of the exposure-indexed Poisson rate formulation.
    Conditioned on the clean image \(X=x\), particle counts follow
    \(Y_\tau\mid X=x\sim\operatorname{Poisson}(\lambda_0\tau x)\).
    The circles indicate count distributions at effective exposures
    \(p\tau\), \(q\tau\), and \(\tau\).
    Poisson thinning decomposes \(Y_\tau\) into complementary views \(A\) and \(B\), where \(q=1-p\), \(A+B=Y_\tau\), and \(A\perp B\mid X\).
    Conditioned on \(X\), these views follow the count distributions associated with effective exposures \(p\tau\) and \(q\tau\), respectively.
    At pixel \(i\), the conditional particle-arrival-rate field component
    \(v_{\tau,i}(y)=\mathbb{E}[\lambda_0X_i\mid Y_\tau=y]\)
    governs transitions toward higher count states.
    Exposure-normalized counts have conditional mean \(\lambda_0x_i\), while their conditional variance decreases with exposure. 
    }
    \label{fig:poisson_path}
\end{figure}

For fixed \(x\), Eq.~\eqref{eq:poisson_count_model_revised} defines the exposure-indexed conditional probability path
\(p_\tau(y\mid x)=\operatorname{Poisson}(y;\lambda_0\tau x)\)
on \(\mathbb{N}_0^n\),
which originates from the degenerate zero-count distribution
\(
p_0(y\mid x)=\mathbf{1}\{y=0\}
\)
at \(\tau=0\).

As exposure increases, particle counts accumulate but are not removed. Conditioned on \(X=x\), \(\{Y_\tau\}_{\tau\geq0}\) therefore forms a pure-birth process on \(\mathbb{N}_0^n\). An arrival at pixel \(i\) drives the transition \(y\rightarrow y+e_i\) at rate \(\lambda_0 x_i\), where \(e_i \in \mathbb{R}^n\) denotes the \(i\)-th standard coordinate basis vector. No reverse transition is present, since the acquisition process only accumulates particle counts.

The conditional path satisfies the Kolmogorov forward equation
\begin{equation}
\partial_\tau p_\tau(y\mid x)
=
\sum_{i=1}^{n}
\lambda_0x_i\bigl(p_\tau(y-e_i\mid x)-p_\tau(y\mid x)\bigr),
\label{eq:conditional_kfe_revised}
\end{equation}
subject to \(p_\tau(y-e_i\mid x)=0\) whenever \(y_i=0\).

These conditional dynamics are governed by the unknown clean rate \(\lambda_0x\). In the following section, we marginalize over the latent-image distribution to obtain a state-dependent particle-arrival-rate field expressed directly in terms of the observed count state.

\subsection{Conditional Particle-Arrival-Rate Field}
\label{sec:conditional_rate}

The conditional dynamics in Sec.~\ref{sec:poisson_probability_path} depend on the unknown clean image \(X\). Marginalizing over the clean-image distribution gives \(p_\tau(y)=\int p_\tau(y\mid x)\,p_{\mathrm{data}}(x)\,dx\).
For every state \(y\) satisfying \(p_\tau(y)>0\), the conditional particle-arrival-rate field from Sec.~\ref{sec:overview},
\(v_\tau(y)=\mathbb{E}[\lambda_0X\mid Y_\tau=y]\),
is again the posterior mean of the latent clean rate given the observed count state, rather than a MAP estimate.

By Bayes' rule,
\begin{equation}
\int\lambda_0x_i\,p_\tau(y\mid x)\,p_{\mathrm{data}}(x)\,dx
=
v_{\tau,i}(y)p_\tau(y).
\end{equation}
Applying this identity while marginalizing the conditional Kolmogorov forward equation~\eqref{eq:conditional_kfe_revised} yields
\begin{equation}
\partial_\tau p_\tau(y)
=
\sum_{i=1}^{n}
\bigl(v_{\tau,i}(y-e_i)p_\tau(y-e_i)-v_{\tau,i}(y)p_\tau(y)\bigr).
\label{eq:marginal_kfe_revised}
\end{equation}
Eq.~\eqref{eq:marginal_kfe_revised} gives $v_{\tau,i}(y)$ a direct
dynamical interpretation as the conditional transition rate for
an additional particle arrival at coordinate $i$. This is the
statistical quantity that PARM aims to estimate: given an observed particle-count state, infer its underlying conditional particle-arrival-rate field. Because
infinitesimal transitions are not directly observable from a
single acquisition, we next construct a finite-exposure target
whose conditional mean equals this transition rate.

\subsection{Finite-Exposure Rate Matching}
\label{sec:finite_exposure_matching}
\label{sec:finite_difference_revised}
\label{sec:gkl_rate_matching}

We convert the transition rate of Sec.~\ref{sec:conditional_rate} into a trainable target by considering a finite-exposure increment.
For \(p\in(0,1)\), the independent-increment property makes
\(\Delta_{p\tau\rightarrow\tau}=(Y_\tau-Y_{p\tau})/((1-p)\tau)\) conditionally unbiased for \(\lambda_0X\), and the tower property gives
\begin{equation}
\mathbb{E}[\Delta_{p\tau\rightarrow\tau}\mid Y_{p\tau}=y]
=
\mathbb{E}[\lambda_0X\mid Y_{p\tau}=y]=v_{p\tau}(y).
\label{eq:finite_difference_target_revised}
\end{equation}
Thus \(\Delta_{p\tau\rightarrow\tau}\) is a finite-exposure surrogate for the same transition-rate field that governs Eq.~\eqref{eq:marginal_kfe_revised}.

To turn the finite-exposure identity in Eq.~\eqref{eq:finite_difference_target_revised}
into a regression objective, we require a loss whose conditional
population risk is minimized by the conditional mean of the rate
target. For nonnegative rate fields, we use the
GKL divergence:
\begin{equation}
D_{\mathrm{GKL}}(u,v)
=
\frac{1}{n}\sum_{i=1}^{n}\bigl(u_i\log(u_i/v_i)-u_i+v_i\bigr),
\label{eq:gkl}
\end{equation}
where \(u\in \mathbb{R}^n_{\ge 0}\) and \(v \in \mathbb{R}^n_{>0}\), \(n\) denotes the dimension of the field, and by definition \(0\log(0/v):=0\).
In the theoretical sections this index set is \(\{1,\ldots,n\}\); in the implementation the same average is taken over a prediction domain \(\Omega\) with \(n=|\Omega|\).
For integer count targets, Eq.~\eqref{eq:gkl} matches the Poisson negative
log-likelihood up to terms independent of the prediction; more generally,
it remains well defined for nonnegative rate targets.

For a nonnegative random target \(U\) with a finite conditional mean and an observation \(Z\), for each observed value \(z\), the conditional expected GKL risk is minimized over positive predictions by
\begin{equation}
v^{*}(z)=\mathbb{E}[U\mid Z=z],
\label{eq:gkl_bayes}
\end{equation}
provided that this conditional mean is strictly positive.
If a component of \(\mathbb{E}[U\mid Z=z]\) is zero, the optimum value is defined as zero under the limiting extension of Eq.~\eqref{eq:gkl}.

If clean images were available, supervised rate matching at effective exposure \(p\tau\) would minimize the population risk
\begin{equation}
\mathcal{R}_{\mathrm{sup}}(f;p\tau)
=
\mathbb{E}\bigl[D_{\mathrm{GKL}}(\lambda_0 X,f(Y_{p\tau}))\bigr],
\label{eq:supervised_rate_risk}
\end{equation}
whose Bayes-optimal predictor is \(f^{*}(y;p\tau)=v_{p\tau}(y)\) by Eq.~\eqref{eq:gkl_bayes}.
The finite-exposure target \(\Delta_{p\tau\rightarrow\tau}\) of Eq.~\eqref{eq:finite_difference_target_revised} has the same conditional mean \(v_{p\tau}(y)\) given \(Y_{p\tau}=y\), so a GKL risk against \(\Delta_{p\tau\rightarrow\tau}\) is minimized by the same predictor.
This establishes the rate-level statistical target for PARM; the corresponding population equivalence for the thinning construction is stated in Section~\ref{sec:self_supervised_thinning}.

\subsection{Self-Supervised Rate Matching by Poisson Thinning}
\label{sec:self_supervised_thinning}

In self-supervised denoising, neither the clean particle-arrival rate \(\lambda_0X\) nor an independently acquired higher-exposure observation is available.
PARM constructs the required lower-exposure state and complementary count increment from a single noisy observation using Poisson thinning.

Given an integer-valued observation \(Y_\tau\), let \(p\in(0,1)\) be a splitting ratio independent of \((X,Y_\tau)\) and set \(q=1-p\).
Poisson thinning yields complementary views
\begin{equation}
A\mid Y_\tau\sim\operatorname{Binomial}(Y_\tau,p),
\qquad
B=Y_\tau-A,
\label{eq:poisson_thinning}
\end{equation}
so that \(A\mid X\sim\operatorname{Poisson}(\lambda_0p\tau X)\), \(B\mid X\sim\operatorname{Poisson}(\lambda_0q\tau X)\), and \(A\perp B\mid X\).
Thus, \(A\) and \(B\) match count observations at effective exposures \(p\tau\) and \(q\tau\), respectively, and either view can serve as the lower-exposure state with the other as its complementary increment.
Equivalently, the thinning construction satisfies \((A,A+B)\mid X\overset{d}{=}(Y_{p\tau},Y_\tau)\mid X\). Hence, the \(A\!\to\!B\) direction realizes exactly the finite-exposure construction of Sec.~\ref{sec:finite_exposure_matching}, while the reverse direction follows symmetrically.

We now state the population equivalence for the thinning construction.
For prediction from \(A\), the complementary counts define the exposure-normalized rate target
\(R^{A}=B/(q\tau)\); the reverse direction uses \(R^{B}=A/(p\tau)\) symmetrically.

\begin{theorem}[Population equivalence of thinning-based rate matching]
\label{thm:finite_difference_equivalence}
Let \(X\sim p_{\mathrm{data}}\) be a nonnegative random image with componentwise finite first moment, and let \(Y_\tau\mid X=x\sim\operatorname{Poisson}(\lambda_0\tau x)\).
For fixed \(p\in(0,1)\), \(q=1-p\), construct \(A,B\) as in Eq.~\eqref{eq:poisson_thinning} and set \(R^{A}=B/(q\tau)\).
For unrestricted measurable predictors with finite GKL risk, consider the supervised clean-rate risk
\[
\mathcal{R}_{\mathrm{sup}}(f;p\tau)=\mathbb{E}\bigl[D_{\mathrm{GKL}}(\lambda_0 X,f(A))\bigr]
\]
and the thinning-based rate-matching risk
\begin{equation}
\mathcal{R}_{\mathrm{self}}(f;p,\tau)=\mathbb{E}\bigl[D_{\mathrm{GKL}}(R^{A},f(A))\bigr],
\label{eq:self_supervised_risk}
\end{equation}
where \(p\) and \(\tau\) are fixed, and both expectations are taken over the joint law of \(X\), \(Y_\tau\mid X\), and \(A\mid Y_\tau\).
Assume that both risks are finite for at least one strictly positive predictor.
Then the two risks share the same Bayes-optimal predictor,
\(f^{*}(a;p\tau)=v_{p\tau}(a)=\mathbb{E}[\lambda_0X\mid A=a]\).
If a component of \(v_{p\tau}(a)\) is zero, the corresponding optimum is zero under the limiting extension of the GKL divergence.
The symmetric statement holds for the \(B\rightarrow A\) direction with \(R^{B}=A/(p\tau)\).
\end{theorem}

\begin{IEEEproof}
Thinning gives \(A\mid X\sim\operatorname{Poisson}(\lambda_0p\tau X)\), \(B\mid X\sim\operatorname{Poisson}(\lambda_0q\tau X)\), and \(A\perp B\mid X\).
Hence \(\mathbb{E}[R^{A}\mid X,A]=\mathbb{E}[B/(q\tau)\mid X,A]=\lambda_0X\).
Taking the conditional expectation given \(A\) and applying the tower property yields
\(\mathbb{E}[R^{A}\mid A]=\mathbb{E}[\lambda_0X\mid A]=v_{p\tau}(A)\).
By Eq.~\eqref{eq:gkl_bayes}, this common conditional mean is the unique Bayes-optimal predictor of both \(\mathcal{R}_{\mathrm{sup}}(\,\cdot\,;p\tau)\) and \(\mathcal{R}_{\mathrm{self}}(\,\cdot\,;p,\tau)\) among unrestricted measurable maps of finite GKL risk.
The zero-component case follows from the limiting extension of Eq.~\eqref{eq:gkl}.
\end{IEEEproof}

Therefore, for every fixed splitting ratio, each thinning direction identifies the conditional particle-arrival-rate field at the effective exposure represented by its input.
PARM randomizes the splitting ratio across training samples and uses both thinning directions, so each network input is drawn from a range of effective exposures.
Concretely, \(p\sim\mathcal{U}(\varepsilon_p,1/2)\) with \(q=1-p\) and a small \(\varepsilon_p>0\).
Because both thinning directions are already used, a pair drawn with \(p\) and one drawn with \(1-p\) differ only by swapping the roles of \(A\) and \(B\), so restricting \(p\) to \((0,1/2]\) loses no training signal.
Unlike flow-matching approaches that condition the network on the flow time, PARM uses the sampled \(p\) only to construct the pair and does not feed it to the network, so the network never observes the effective exposure of its input.
We refer to this as the \emph{implicit-exposure} setting.

To characterize it, let \(J\in\{A,B\}\) denote the two equally weighted directions and define the triple
\begin{equation}
(Z,U,\tau_Z)=
\begin{cases}
(A,\,R^{A},\,p\tau), & J=A,\\[2pt]
(B,\,R^{B},\,q\tau), & J=B,
\end{cases}
\label{eq:direction_triple}
\end{equation}
where \(Z\) is the network input, \(U\) the rate target, and \(\tau_Z\) the effective exposure.
Conditioned on \(\tau_Z\) and \(J\), Theorem~\ref{thm:finite_difference_equivalence} gives \(\mathbb{E}[U\mid Z=z,\tau_Z,J]=\mathbb{E}[\lambda_0X\mid Z=z,\tau_Z,J]\).
Since the network observes only \(Z\), the law of iterated expectation marginalizes \(\tau_Z\) and \(J\) out, yielding \(\mathbb{E}[U\mid Z=z]=\mathbb{E}[\lambda_0X\mid Z=z]\).
Thus, for exact Poisson counts, the population equivalence is preserved under the direction-weighted randomized-thinning distribution, and the Bayes-optimal target is the exposure-marginalized predictor \(\mathbb{E}[\lambda_0X\mid Z=z]\).
This is the implicit-exposure implementation used by PARM, with a single shared network trained on both directions as a bidirectional multitask parameterization.

The equivalence above is exact for calibrated integer-valued Poisson observations.
When detector measurements may not be integer-valued, PARM applies count-unit calibration and expectation-preserving randomized rounding before thinning (Supplementary Material, Sec.~I-B).
In that case, the construction preserves the first moment, but the exact Poisson-splitting and population-equivalence statements should be interpreted as approximations.
Further theoretical statements are collected in Appendix~\ref{app:additional_theory}; the GKL conditional-mean argument and the expectation in Theorem~\ref{thm:finite_difference_equivalence} are given in the Supplementary Material, Sec.~II.
The practical network parameterization, training objective, and low-count stabilization are described in Sec.~\ref{sec:practical_allocation}.

\subsection{Practical Rate-Allocation Learning and Low-Count Stabilization}
\label{sec:practical_allocation}
\label{sec:soft_rate_consistency_revised}
\label{sec:training_objective}

Theorem~\ref{thm:finite_difference_equivalence} identifies the conditional particle-arrival rate.
The global magnitude of the rate depends on the acquisition exposure and count-unit calibration.
PARM therefore learns the normalized spatial allocation of the conditional rate and recovers its global particle mass separately during inference.
Accordingly, we decompose
\begin{equation}
m_\tau(y)=\sum_{j=1}^{n}v_{\tau,j}(y),
\qquad
\pi_{\tau,i}^{*}(y)=\frac{v_{\tau,i}(y)}{\sum_{j=1}^{n}v_{\tau,j}(y)},
\label{eq:rate_allocation_decomposition}
\end{equation}
whenever \(m_\tau(y)>0\), where \(\pi_\tau^{*}\) is the Bayes-optimal normalized allocation target.

Let \(h_\theta(Z)\in\mathbb{R}^{|\Omega|}\) denote the unconstrained network output for an input count field \(Z\), where \(\Omega\) is the prediction domain: the spatial domain for 2D inputs and the complete spatiotemporal output volume for 3D inputs.
This \(\Omega\) is the implementation counterpart of the index set \(\{1,\ldots,n\}\) used above, with \(n=|\Omega|\).
The predicted allocation is the softmax normalization
\begin{equation}
\pi_{\theta,i}(Z)
=
\frac{\exp(h_{\theta,i}(Z))}{\sum_{j\in\Omega}\exp(h_{\theta,j}(Z))},
\qquad i\in\Omega,
\label{eq:main_softmax_allocation}
\end{equation}
which maps $h_\theta(Z)$ to the probability simplex over $\Omega$.
Under a GKL risk with target $\mathbb{E}[\lambda_0X\mid Z=z]$, the population-optimal allocation is
\begin{equation}
\pi_i^{*}(z)
=
\frac{\mathbb{E}[\lambda_0X_i\mid Z=z]}{\sum_{j\in\Omega}\mathbb{E}[\lambda_0X_j\mid Z=z]}.
\end{equation}
Because the splitting ratio is not provided to the network, the conditional mean is marginalized over the exposure states induced by randomized thinning.
The corresponding GKL mass--allocation factorization is given in Appendix~\ref{app:rate_allocation}.

Rather than dividing by the sampled exposure fractions that define \(R^{A}\) and \(R^{B}\), the implementation applies empirical unit-mean normalization within each minibatch.
The practical bidirectional targets are \(\widetilde R^{A}=B/\overline B\) and \(\widetilde R^{B}=A/\overline A\),
where $\overline A$ and $\overline B$ denote the corresponding minibatch means, and $\widetilde R^{A}$ supervises the prediction from $A$.
This preserves relative allocation while removing a common empirical scale, which is an implementation-level approximation to the ideal exposure-normalized targets.

The bidirectional softmax--GKL objective averages the two complementary directions,
\begin{equation}
\mathcal{L}_{\mathrm{bi}}(\theta)
=
-\mathbb{E}\!\left[
\frac{1}{2|\Omega|}
\sum_{i\in\Omega}
\bigl(\widetilde R_i^{A}\log \pi_{\theta,i}(A)+\widetilde R_i^{B}\log \pi_{\theta,i}(B)\bigr)
\right].
\label{eq:bidirectional_allocation_loss}
\end{equation}
At extremely low counts, complementary thinning can produce substantially different sparse realizations, even though both views originate from the same latent scene.
We therefore add a soft cross-view consistency regularization on the raw network outputs,
\begin{equation}
\mathcal{R}_{\mathrm{cons}}(\theta)
=
\mathbb{E}\!\left[
\frac{1}{|\Omega|}\left\|h_\theta(A)-h_\theta(B)\right\|_1
\right].
\label{eq:rate_consistency_revised}
\end{equation}
The complete training objective is
\begin{equation}
\mathcal{L}_{\mathrm{total}}(\theta)
=
\mathcal{L}_{\mathrm{bi}}(\theta)+\gamma(\widehat{\Lambda}_{\mathrm{occ}})\mathcal{R}_{\mathrm{cons}}(\theta),
\label{eq:total_loss_revised}
\end{equation}
where $\widehat p_0$ is the fraction of zero-valued sites in the complete raw observation, $\widehat{\Lambda}_{\mathrm{occ}}=-\log\widehat p_0$ is a global occupancy statistic, and the consistency weight is
\begin{equation}
\gamma(\widehat{\Lambda}_{\mathrm{occ}})
=
\gamma_0\frac{\Lambda_{\mathrm{th}}}{\widehat{\Lambda}_{\mathrm{occ}}}\,
\mathbf{1}\{\widehat{\Lambda}_{\mathrm{occ}}<\Lambda_{\mathrm{th}}\},
\label{eq:gamma_schedule_revised}
\end{equation}
with $\gamma_0=0.01$ and $\Lambda_{\mathrm{th}}=0.4$.
When zeros are abundant, $\widehat p_0$ is large, $\widehat{\Lambda}_{\mathrm{occ}}$ is small, and consistency is active; when zeros are scarce, $\widehat p_0$ is small, $\widehat{\Lambda}_{\mathrm{occ}}$ is large, and the indicator disables consistency.
If $\widehat p_0=0$, then $-\log\widehat p_0=+\infty$, which likewise turns consistency off.
The threshold $\Lambda_{\mathrm{th}}=0.4$ is an empirical operating point: at $\Lambda=0.4$, the conditional multi-event fraction $\Pr(N\geq2\mid N>0)$ is about $20\%$.
Thus $\widehat{\Lambda}_{\mathrm{occ}}$ is a global sparsity indicator, and the consistency term acts as a weak low-count stabilizer.
The following section describes how the learned allocation is combined with the observed global particle mass to obtain the final reconstruction.

\subsection{Inference and Observed-Mass-Preserving Reconstruction}
\label{sec:mass_preserving_reconstruction}

The rate-allocation formulation separates two quantities: the relative distribution of particle arrivals, inferred by the network, and the global particle mass, obtained from the observation.
This separation allows PARM to restore an acquisition-scale image without the network predicting an absolute exposure-dependent magnitude.

From Sec.~\ref{sec:conditional_rate}, the conditional rate field is \(v_\tau(y)=\mathbb{E}[\lambda_0X\mid Y_\tau=y]\), whose normalized form is the Bayes-optimal allocation \(\pi_\tau^{*}(y)\) of Eq.~\eqref{eq:rate_allocation_decomposition}.
The posterior mean of the clean count field satisfies \(\mathbb{E}[\lambda_0\tau X\mid Y_\tau]=\tau v_\tau(Y_\tau)\), so the multiplicative factor \(\tau\) cancels from this normalization. Thus, the allocation is invariant to a common rate--count scaling, although its conditional dependence on the observation distribution may still vary with exposure.
When exposure is not provided as a network input, the trained allocation corresponds to the exposure-marginalized predictor of Sec.~\ref{sec:self_supervised_thinning}.
Theorem~\ref{thm:finite_difference_equivalence} characterizes inputs drawn from the randomized thinning distribution, so applying the trained network to the complete observation \(Y_\tau\) is a practical inference rule rather than a direct consequence of the theorem.

In the reported experiments, the inference allocation aggregates \(K=3\) network evaluations: the original input and sequentially thinned versions are each processed by the network, and their positive candidate fields are combined by empirical score-weighted aggregation (Supplementary Material, Sec.~I-D); this rule is not a Bayesian posterior.
Let \(V_\theta^{\mathrm{inf}}(Y_\tau)>0\) denote the resulting positive, unnormalized field after this aggregation.
Let \(\Omega\) denote the complete reconstructed spatial or spatiotemporal domain.
We normalize the aggregated positive field over \(\Omega\):
\begin{equation}
\pi_{\theta,i}^{\mathrm{inf}}(Y_\tau)
=
\frac{V_{\theta,i}^{\mathrm{inf}}(Y_\tau)}{\sum_{j\in\Omega}V_{\theta,j}^{\mathrm{inf}}(Y_\tau)},
\qquad i\in\Omega.
\label{eq:inference_allocation}
\end{equation}
By construction, $\sum_{i\in\Omega}\pi_{\theta,i}^{\mathrm{inf}}(Y_\tau)=1$.

The final PARM reconstruction distributes the observed total count according to this allocation.
Let \(\widehat{M}_\tau=\sum_{j\in\Omega}Y_{\tau,j}\) denote the observed particle mass over \(\Omega\). Then
\begin{equation}
\widehat{C}_{\tau,i}=\widehat{M}_\tau\,\pi_{\theta,i}^{\mathrm{inf}}(Y_\tau),
\label{eq:mass_allocation_reconstruction}
\end{equation}
so that \(\sum_{i\in\Omega}\widehat{C}_{\tau,i}=\widehat{M}_\tau\): the network determines how the observed particle mass is distributed, while the measurement determines its global magnitude.
Poisson additivity makes \(\widehat{M}_\tau\) a conditionally unbiased estimate of the latent total count \(M_\tau=\sum_{i\in\Omega}\lambda_0\tau X_i\), and because it is aggregated over \(\Omega\), its relative fluctuation is smaller than that of individual low-count sites.

This observed-mass-preserving reconstruction is a mass--allocation factorization, not an exact Bayes estimator: \(\widehat{M}_\tau\) is conditionally unbiased but is not the posterior mean, so when \(\widehat{M}_\tau=0\) the reconstruction is identically zero even though the posterior mean can be positive.

When $\Lambda_{\mathrm{acq}}=\lambda_0\tau$ is known, the exposure-independent latent intensity estimate is $\widehat{X}=\widehat{C}_\tau/\Lambda_{\mathrm{acq}}$.
The equivalence result assumes calibrated integer-valued Poisson counts, known exposure intervals, and unrestricted predictors optimized under GKL; the consistency regularizer and observed-mass-preserving reconstruction are practical components outside the theorem's guarantee (Appendix~\ref{app:theoretical_scope}).

\section{Experiments}
\label{sec:experiments}

We evaluate PARM in three settings:
(i) a controlled synthetic Cameraman benchmark with known ground truth,
(ii) standard natural videos corrupted by synthetic Poisson noise, and
(iii) real low-dose electron microscopy without clean ground truth.
For the first two settings, the clean source data permit quantitative evaluation using peak signal-to-noise ratio (PSNR) and structural similarity index (SSIM).
For electron microscopy, we assess structural preservation and agreement with local multi-frame, high-SNR references used only as references, not as ground truth.

We compare PARM against three categories of baselines, including model-based frameworks, single-image deep learning methods, and sequence-based deep learning approaches.
The model-based methods are BM3D~\cite{dabov2007image}, VBM4D~\cite{maggioni2012vbm4d}, and PURE-LET~\cite{luisier2010fast}.
The single-image deep learning methods include Noise2Void (N2V)~\cite{krull2019noise2void} and Pixel2Pixel~\cite{ma2025pixel2pixel}, while the sequence-based methods are SHINE~\cite{shine}, UDVD~\cite{udvd}, and bit2bit~\cite{bit2bit}.
Methods without a native volumetric or video formulation (such as BM3D, PURE-LET, and Pixel2Pixel) are evaluated only in the 2D setting, and we do not introduce new 3D backbones solely for comparison.

\subsection{Implementation Details and Baseline Settings}
\label{sec:implementation}

PARM uses a 2D U-Net for single-frame experiments and a 3D U-Net for sequence-based experiments.
The corresponding network input is a noisy count image or spatiotemporal count sequence, and neither the clean image nor the exposure level is provided as an additional input.
For the synthetic Cameraman benchmark, the 3D network input size is \(32\times64\times64\) (\(T\times H\times W\)).

PARM is optimized separately for each evaluated sequence and particle level. The 2D and 3D networks are initialized from their corresponding architecture-specific checkpoints pretrained on the noisy Flower sequence using the same self-supervised PARM objective.
No clean frames or ground-truth targets are used during pretraining. Following initialization, all network parameters are optimized on the noisy target sequence at the corresponding particle level. Thus, the final models are specific to each sequence and particle level, though they share a fixed self-supervised initialization. Because Flower is included in Table~\ref{table_noisy_videos}, its result corresponds to same-sequence self-supervised warm-starting rather than evaluation on a previously unseen sequence.

We preserve the original training settings of each baseline whenever possible.
Sequence-specific self-supervised training is used only when supported by the original method.
For N2V, SHINE, and bit2bit, we follow the recommended configurations and change only input dimensions or noise-related parameters required for compatibility with the evaluated data.
UDVD is evaluated with its original default configuration, while BM3D and VBM4D use their standard settings.
No clean targets are used for PARM or any self-supervised baseline.
PARM(Sup.) is trained with clean targets and is included only as an empirical supervised reference.
Additional qualitative comparisons on Cameraman, natural video, and electron microscopy are provided in the Supplementary Material.

\subsection{Synthetic Poisson Denoising on Cameraman}
\label{sec:simulation}

\subsubsection{Dataset and setup}

We construct a controlled benchmark from the \(256\times256\) Cameraman image.
A spatial trajectory shaped as the digit "2" is defined on the image, and \(128\times128\) patches sampled along this trajectory form an ordered clean stack.
We then corrupt the clean patches with Poisson shot noise at multiple particle levels.
The paired clean and noisy stacks permit direct quantitative evaluation under a controlled particle-counting model.

We evaluate 2D and 3D variants of PARM using the same self-supervised rate-matching objective.
The 2D model processes each noisy image independently, whereas the 3D model jointly processes ordered images to use their spatiotemporal context.
We also report PARM(Sup.), trained with clean targets, solely as an empirical supervised reference.

\begin{figure*}[t]
 \centering
 \includegraphics[width=1.0\textwidth]{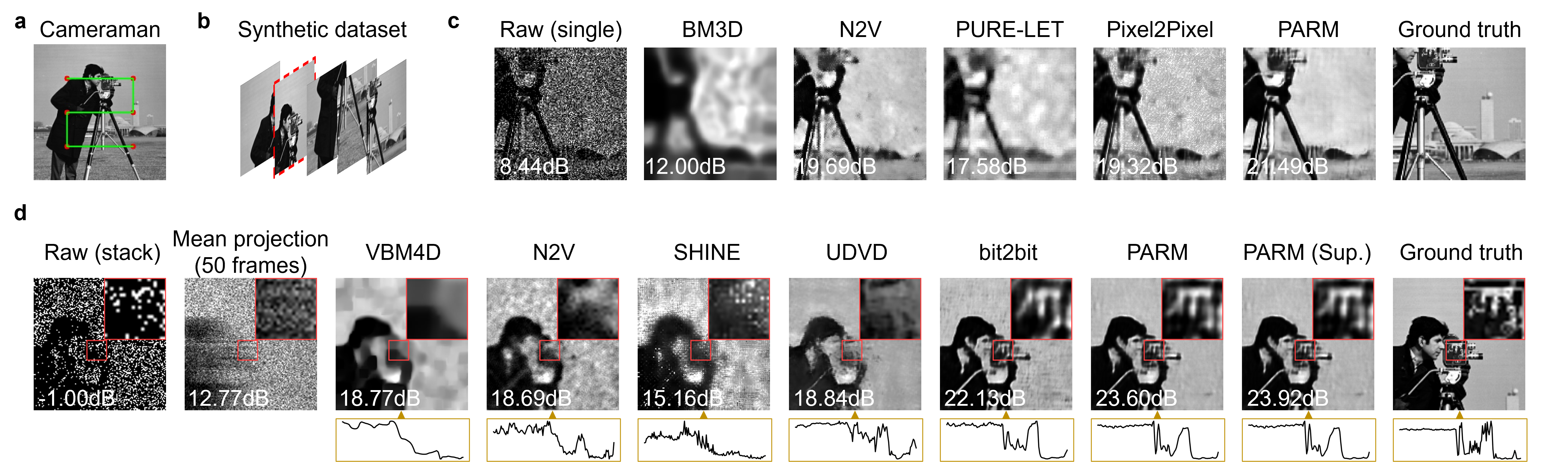}
 \caption{
 Synthetic Cameraman benchmark for 2D and 3D Poisson denoising.
 (a) Original Cameraman image and the "2"-shaped cropping trajectory used to generate the synthetic stack.
 (b) Illustration of the resulting ordered synthetic dataset, where consecutive patches contain smoothly varying spatial content.
 (c) 2D single-frame Poisson denoising comparison at $\Lambda=2$.
 PARM suppresses particle noise while retaining object boundaries and local intensity transitions.
 (d) 3D stack-based Poisson denoising comparison at $\Lambda=0.25$.
 The red box highlights the camera.
 PARM shows lower visible particle fluctuations and sharper local transitions in this region.
 The magnified insets and line profiles further show the preservation of local transitions and structural peaks.
 PSNR values are annotated on the reconstructions in (c) and (d).
 }
 \label{fig:cameraman}
\end{figure*}

\subsubsection{Quantitative results}

Table~\ref{table_cameraman} reports PSNR and SSIM averaged over all synthetic image frames at each particle level.
PARM achieved the highest PSNR and SSIM among the self-supervised methods at every tested particle level.
Its advantage was most evident at low counts, where the compared methods retained particle noise or oversmoothed fine structures.
PARM(Sup.) remained stronger, as expected from clean supervision, but PARM approached this supervised reference without using clean targets.

\begin{table*}[t]
\caption{
Quantitative comparison for the 3D stack-based Cameraman experiment.
PSNR and SSIM are averaged over all extracted images at each particle level. 
Bold and underlined values indicate the best and second-best results, respectively, among all methods except PARM(Sup.). PARM(Sup.) is shown only as a supervised reference and is excluded from this ranking. Results are averaged over 5 independent noise realizations.
}
\label{table_cameraman}
 \centering
 \resizebox{\textwidth}{!}{
 \tiny
 \begin{tabular}{||c||c|c|c|c|c|c||}
 \hline
 $\Lambda$ & 2 & 1 & 0.5 & 0.25 & 0.125 & 0.0625 \\
 \hline
 Method & \multicolumn{6}{|c||}{PSNR (dB) / SSIM} \\
 \hline
 VBM4D & 25.08 / \underline{0.796} & 22.71 / \underline{0.732} & 20.46 / 0.643 & 18.72 / 0.549 & 17.15 / 0.450 & 15.67 / 0.361 \\
N2V & 20.96 / 0.527 & 20.12 / 0.498 & 18.99 / 0.435 & 18.27 / 0.383 & 16.98 / 0.298 & 16.16 / 0.238 \\
SHINE & 18.13 / 0.405 & 17.84 / 0.400 & 16.91 / 0.376 & 15.64 / 0.292 & 14.43 / 0.215 & 12.61 / 0.121 \\
UDVD & 22.93 / 0.736 & 21.82 / 0.700 & 20.63 / 0.657 & 19.40 / 0.587 & 18.34 / 0.524 & 17.37 / 0.450 \\
bit2bit & \underline{25.40} / 0.781& \underline{23.24} / 0.703& \underline{22.35} / \underline{0.683}& \underline{22.03} / \underline{0.678}& \underline{20.49} / \underline{0.582}& \underline{19.01} / \underline{0.487}\\
PARM & \textbf{26.12} / \textbf{0.817}& \textbf{24.79} / \textbf{0.787}& \textbf{23.39} / \textbf{0.731}& \textbf{22.28} / \textbf{0.708}& \textbf{20.99} / \textbf{0.646}& \textbf{19.46} / \textbf{0.588}\\
\hline
 PARM(Sup.) & {26.42} / {0.832} & {24.99} / {0.794} & {23.78} / {0.762} & {22.84} / {0.736} & {21.31} / {0.677} & {19.75} / {0.619} \\
 \hline
 \end{tabular}
 }
\end{table*}

\subsubsection{Qualitative analysis}

Fig.~\ref{fig:cameraman} presents representative 2D and 3D reconstructions.
In the 2D comparison in Fig.~\ref{fig:cameraman}(c), PARM reduced noise while retaining the tripod, body contour, and local edges.
In the 3D comparison in Fig.~\ref{fig:cameraman}(d), PARM recovered the camera region with clearer contrast and shape.
The compared methods either blurred this region into the background or retained substantial particle noise.
The vertical line profiles show sharper intensity variations for PARM, consistent with Table~\ref{table_cameraman}.

\subsection{Poisson Denoising on Natural Video}
\label{sec:video_denoising}

We next evaluate PARM on six standard grayscale video sequences: \textit{Coastguard}, \textit{Bus}, \textit{Flower}, \textit{Foreman}, \textit{Miss America}, and \textit{Tennis}.
Each clean frame is corrupted by synthetic Poisson noise at multiple particle levels.
We measure reconstruction quality against the original clean video using PSNR and SSIM.
For PARM and stack-compatible self-supervised baselines, the input size is \(32\times128\times128\) (\(T\times H\times W\)).

\begin{figure}[t]
 \centering
 \includegraphics[width=1.0\linewidth]{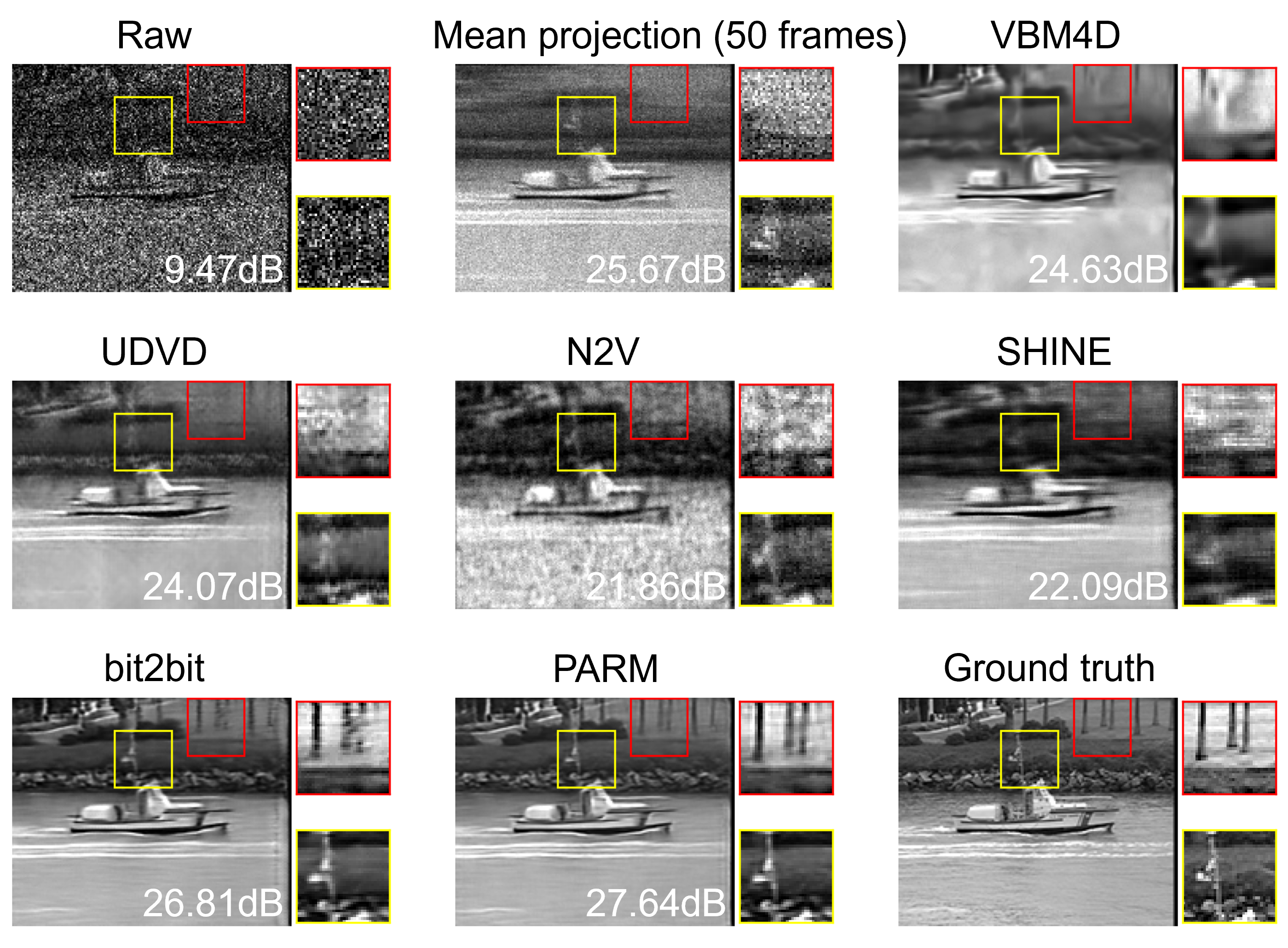}
 \caption{
 Qualitative comparison on the \textit{Coastguard} sequence under synthetic Poisson noise at $\Lambda=2$.
 Red and yellow boxes mark fine-structure regions, with enlarged views shown beside each reconstruction.
 In the red region, PARM resolves the thin tree trunks more clearly, whereas other methods merge or obscure them.
 In the yellow region, PARM recovers the sail with a recognizable shape and boundary.
 }
 \label{fig:noisy_videos}
\end{figure}

Table~\ref{table_noisy_videos} reports results for all six sequences.
PARM ranked first or second for most combinations of sequence, particle level, and metric, although competing methods performed better in several settings.
The strongest gains occurred in scenes where thin structures and local textures were difficult to retain.
In Fig.~\ref{fig:noisy_videos}, these gains correspond to clearer tree trunks in the red region and a more recognizable sail in the yellow region.

\begin{table*}[!htbp]
\caption{
Quantitative comparison on six noisy video sequences.
PSNR (dB) and SSIM are reported for each particle level $\Lambda$.
Bold numbers indicate the best result for each metric and setting, and underlined numbers indicate the second-best result. Results are averaged over 5 independent noise realizations.
}
\label{table_noisy_videos}
 \centering
 \tiny
 \resizebox{\textwidth}{!}{
 \begin{tabular}{||c||c|c|c|c|c|c||}
 \hline
 Method & VBM4D & N2V & SHINE & UDVD & bit2bit & PARM \\
 \hline
 & \multicolumn{6}{|c||}{Coastguard} \\
 \hline
 2 & 24.84 / 0.628 & 22.47 / 0.502 & 22.75 / 0.599 & 24.10 / 0.604 & \underline{26.40} / \underline{0.746} & \textbf{27.25} / \textbf{0.767}\\
 1 & 23.24 / 0.537 & 19.98 / 0.342 & 21.24 / 0.513 & 23.32 / 0.563 & \underline{25.18} / \underline{0.701} & \textbf{26.26} / \textbf{0.723}\\
 0.5 & 21.55 / 0.447 & 17.55 / 0.193 & 18.81 / 0.432 & 22.57 / 0.518 & \underline{24.58} / \underline{0.654} & \textbf{25.12} / \textbf{0.666} \\
 0.25 & 20.17 / 0.375 & 15.65 / 0.122 & 17.61 / 0.353 & 21.72 / 0.472 & \underline{23.07} / \underline{0.543} & \textbf{23.91} / \textbf{0.592} \\
 0.125 & 18.75 / 0.316 & 10.94 / 0.035 & 17.39 / 0.320 & \underline{20.84} / \underline{0.416} & 19.57 / 0.327 & \textbf{22.33} / \textbf{0.471} \\
 0.0625 & 16.86 / 0.248 & 8.64 / 0.021 & 17.32 / 0.289 & \underline{19.71} / \underline{0.339} & 19.48 / 0.325 & \textbf{20.68} / \textbf{0.387} \\
 \hline
 & \multicolumn{6}{|c||}{Bus} \\
 \hline
 2 & 24.06 / 0.676 & 21.33 / 0.553 & 23.11 / 0.678 & 21.80 / 0.567 & \underline{24.36} / \underline{0.709} & \textbf{25.22} / \textbf{0.737} \\
 1 & 22.20 / 0.578 & 20.16 / 0.465 & 22.18 / 0.631 & 21.17 / 0.524 & \underline{23.11} / \underline{0.651} & \textbf{24.01} / \textbf{0.687} \\
 0.5 & 20.45 / 0.472 & 19.14 / 0.371 & 21.70 / \underline{0.590} & 20.46 / 0.471 & \underline{21.99} / 0.588 & \textbf{22.79} / \textbf{0.627} \\
 0.25 & 19.24 / 0.397 & 17.98 / 0.290 & 20.51 / \underline{0.530} & 19.81 / 0.422 & \underline{20.82} / 0.502 & \textbf{21.57} / \textbf{0.555} \\
 0.125 & 18.34 / 0.346 & 17.09 / 0.208 & 19.56 / \underline{0.456} & 19.31 / 0.387 & \underline{20.13} / \underline{0.456} & \textbf{20.43} / \textbf{0.480} \\
 0.0625 & 17.27 / 0.298 & 15.64 / 0.143 & 18.18 / 0.383 & 18.84 / 0.353 & \underline{19.32} / \underline{0.394} & \textbf{19.55} / \textbf{0.419} \\
 \hline
 & \multicolumn{6}{|c||}{Flower} \\
 \hline
 2 & 19.41 / 0.642 & 15.45 / 0.281 & 18.93 / 0.664 & 16.87 / 0.465 & \underline{20.22} / \underline{0.700} & \textbf{21.08} / \textbf{0.750} \\
 1 & 16.92 / 0.444 & 14.92 / 0.223 & 17.55 / 0.559 & 16.23 / 0.395 & \underline{18.93} / \underline{0.638} & \textbf{19.83} / \textbf{0.688} \\
 0.5 & 15.54 / 0.302 & 13.82 / 0.154 & 16.37 / 0.456 & 15.71 / 0.325 & \underline{17.89} / \underline{0.563} & \textbf{18.58} / \textbf{0.620} \\
 0.25 & 14.94 / 0.241 & 12.33 / 0.088 & 15.15 / 0.360 & 15.13 / 0.241 & \underline{16.78} / \underline{0.445} & \textbf{17.42} / \textbf{0.516} \\
 0.125 & 14.36 / 0.198 & 9.91 / 0.029 & 13.67 / 0.287 & 14.69 / 0.176 & \underline{16.17} / \underline{0.379} & \textbf{16.33} / \textbf{0.383} \\
 0.0625 & 13.42 / 0.152 & 9.09 / 0.027 & 11.96 / 0.194 & 14.43 / 0.161 & \textbf{15.41} / \textbf{0.268} & \underline{15.36} / \underline{0.265} \\
 \hline
 & \multicolumn{6}{|c||}{Foreman} \\
 \hline
 2 & \underline{25.89} / 0.694 & 23.74 / 0.555 & 25.70 / \underline{0.713} & 23.81 / 0.603 & 24.70 / 0.676 & \textbf{26.15} / \textbf{0.719} \\
 1 & \underline{24.37} / 0.637 & 22.60 / 0.474 & 23.66 / \underline{0.666} & 23.12 / 0.578 & 23.31 / 0.617 & \textbf{24.72} / \textbf{0.682} \\
 0.5 & 22.74 / 0.575 & 21.33 / 0.391 & 19.43 / 0.596 & 22.37 / 0.543 & \underline{23.02} / \underline{0.609} & \textbf{23.76} / \textbf{0.646}\\
 0.25 & 21.29 / 0.509 & 19.57 / 0.256 & 12.82 / 0.458 & 21.44 / 0.500 & \underline{21.98} / \underline{0.567} & \textbf{22.21} / \textbf{0.607}\\
 0.125 & 19.00 / 0.413 & 17.53 / 0.168 & 12.28 / 0.416 & 20.73 / 0.472 & \underline{21.06} / \underline{0.525} & \textbf{21.10} / \textbf{0.571}\\
 0.0625 & 16.19 / 0.304 & 16.48 / 0.131 & 12.63 / 0.260 & \underline{20.00} / 0.449 & \textbf{20.38} / \underline{0.479} & 19.67 / \textbf{0.535}\\
 \hline
 & \multicolumn{6}{|c||}{Miss America} \\
 \hline
 2 & \textbf{32.36} / \underline{0.858} & 28.02 / 0.700 & 19.01 / 0.796 & 29.52 / 0.802 & 30.83 / 0.833 & \underline{31.69} / \textbf{0.863} \\
 1 & \underline{30.52} / \underline{0.833} & 26.87 / 0.638 & 18.79 / 0.773 & 28.92 / 0.787 & 30.34 / 0.829 & \textbf{30.88} / \textbf{0.849}\\
 0.5 & 28.39 / 0.802 & 25.49 / 0.562 & 21.49 / 0.784 & 28.04 / 0.769 & \textbf{29.93} / \underline{0.826} & \underline{29.89} / \textbf{0.833}\\
 0.25 & 26.72 / 0.756 & 23.79 / 0.441 & 21.86 / 0.758 & 27.11 / 0.748 & \underline{28.75} / \underline{0.805} & \textbf{29.08} / \textbf{0.817}\\
 0.125 & 24.63 / 0.675 & 21.26 / 0.294 & 23.99 / 0.738 & 26.19 / 0.730 & \underline{27.64} / \underline{0.780} & \textbf{28.05} / \textbf{0.801}\\
 0.0625 & 21.83 / 0.560 & 19.68 / 0.179 & 17.88 / 0.638 & 25.25 / 0.706 & \underline{26.19} / \underline{0.745} & \textbf{26.32} / \textbf{0.779}\\
 \hline
 & \multicolumn{6}{|c||}{Tennis} \\
 \hline
 2 & \underline{24.17} / \underline{0.590} & 18.33 / 0.263 & 23.38 / 0.585 & 20.47 / 0.467 & 23.09 / 0.568 & \textbf{24.74} / \textbf{0.614} \\
 1 & 21.98 / 0.524 & 17.45 / 0.185 & \underline{22.21} / \underline{0.543} & 19.85 / 0.431 & 21.85 / 0.526 & \textbf{23.65} / \textbf{0.579} \\
 0.5 & 19.76 / 0.438 & 16.09 / 0.090 & 20.75 / \underline{0.500} & 19.40 / 0.403 & \underline{21.12} / 0.491 & \textbf{22.38} / \textbf{0.539} \\
 0.25 & 18.77 / 0.381 & 15.79 / 0.103 & 18.90 / \underline{0.440} & 18.92 / 0.379 & \underline{19.84} / 0.417 & \textbf{21.09} / \textbf{0.491} \\
 0.125 & 17.73 / 0.317 & 11.98 / 0.035 & 14.79 / 0.359 & 18.54 / 0.362 & \underline{19.04} / \underline{0.377} & \textbf{19.52} / \textbf{0.438} \\
 0.0625 & 15.87 / 0.237 & 9.45 / 0.025 & 11.85 / 0.263 & 18.19 / \underline{0.357} & \underline{18.49} / 0.328 & \textbf{18.52} / \textbf{0.397} \\
 \hline
 \end{tabular}
 }
\end{table*}
\subsection{Real Low-Dose Electron Microscopy}
\label{sec:real_em}

Finally, we evaluate PARM on a real low-dose electron microscopy movie of
platinum (Pt) nanoparticles supported on ceria and acquired with a direct electron
detector~\cite{em}. The movie was recorded at an electron dose rate of
approximately $2000\,\text{e}^- \text{\AA}^{-2}\text{s}^{-1}$ with a
frame rate of 75 frames per second (fps), corresponding to a 13-ms exposure and approximately
$26\,\text{e}^- \text{\AA}^{-2}$ per frame (${\sim}0.2\,\text{e}^-$ per
pixel). The single-frame signal-to-noise ratio (SNR) measured in vacuum is approximately 0.45,
indicating a strongly dose-limited regime.

Because clean ground truth is unavailable, we assess reconstruction
consistency against group-averaged measurements at different temporal
scales. For a group size $G$, the noisy and denoised stacks
are partitioned into non-overlapping groups of $G$ consecutive frames.
Within each group, the noisy and denoised frames are separately averaged
to obtain a group-averaged noisy reference and a group-averaged denoised
reconstruction. We compute the Pearson correlation coefficient between these two
images and average the resulting correlations
over all temporal groups. Increasing $G$ suppresses stochastic
frame-to-frame fluctuations in the noisy measurement, producing a
progressively higher-SNR reference for assessing structural consistency.
These group-averaged references are used only for
evaluation and are never used during network training.
Because each reference is formed from the same measured sequence, this
correlation assesses consistency rather than fidelity to an unknown clean image.

To assess spatial variability, we additionally separate the original
stack into four non-overlapping spatially subsampled stacks using
the four offsets of $2\times2$ downsampling. The complete denoising and
correlation analyses are performed independently on each spatial subset.
In Fig.~\ref{fig:Pt_em}(c), the bar height represents the mean correlation
across the four spatial subsets, while the error bars indicate the
corresponding standard deviation.

\begin{figure*}[!htbp]
 \centering
 \includegraphics[width=1.0\linewidth]{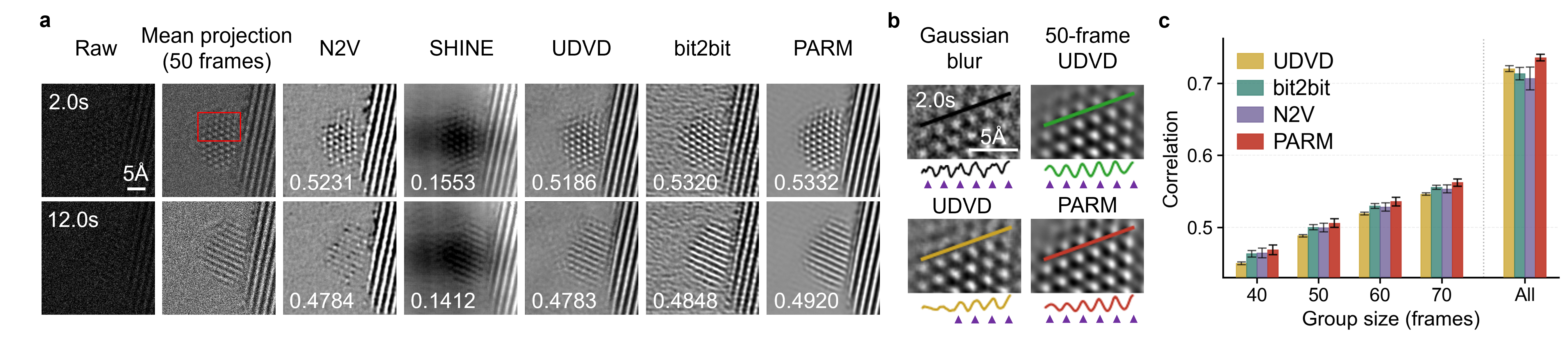}
 \caption{
 Real low-dose electron microscopy denoising on Pt nanoparticle data.
 (\textbf{a}) Representative results at two time points.
 Numbers on the denoised images denote correlations with a local multi-frame reference.
 PARM suppresses background noise while retaining lattice-like nanoparticle structures.
 (\textbf{b}) Detailed comparison of local nanoparticle structures and intensity profiles.
 The region of interest in (\textbf{b}) corresponds to the red box in (\textbf{a}).
 Gaussian blur and 50-frame UDVD provide reference comparisons.
 Gaussian blur suppresses noise but removes particle boundaries, whereas 50-frame UDVD uses longer temporal context but blurs some single-frame details.
 PARM produces a cleaner single-frame reconstruction while retaining individual nanoparticle contrast.
 Its intensity profile preserves periodic peaks that the reference methods attenuate or merge.
 (\textbf{c}) Average correlation as a function of the group size used to construct the evaluation reference. The displayed correlations in (\textbf{a}) use $G=50$. Error bars in (\textbf{c}) denote the sample standard deviation across the four $2\times2$ spatial-subsampling offsets; they do not represent independently acquired microscopy experiments.
 }
 \label{fig:Pt_em}
\end{figure*}

Fig.~\ref{fig:Pt_em}(a) presents representative reconstructions at two time points.
Compared with N2V, SHINE, UDVD, and bit2bit, PARM produced cleaner backgrounds while retaining the periodic nanoparticle structure.
The local comparison in Fig.~\ref{fig:Pt_em}(b) further distinguishes structural preservation from smoothing.
Gaussian blur reduced high-frequency noise but removed nanoparticle boundaries.
The 50-frame UDVD reconstruction used longer temporal context but blurred some single-frame details.
PARM retained clearer nanoparticle contrast without heavy temporal averaging.
Its intensity profiles preserved visible nanoparticle peaks that the reference methods weakened, broadened, or merged.
Fig.~\ref{fig:Pt_em}(c) reports average correlation for different reference group sizes.
PARM achieved the highest measured correlation under the evaluated consistency protocol.
Together, these results support its applicability to real low-dose microscopy without clean supervision.

\subsection{Ablation Study}
\label{sec:ablation}

We analyze three practical factors: the base consistency coefficient, the
temporal stack size, and the thinning-based inference aggregation mode.
Fig.~\ref{fig:ablation} summarizes the consistency-coefficient ablation
and the stack-size ablation on the 3D Cameraman benchmark.

The implemented regularization strength follows the occupancy-dependent
schedule in Eq.~\eqref{eq:gamma_schedule_revised}.
On Cameraman, a large base coefficient degrades the displayed low-count
settings.
The Coastguard sweep in the Supplementary Material, Sec.~IV-B, shows a
clearer cross-regime tradeoff: strong consistency helps only in the
sparsest regimes and reduces PSNR at higher particle levels.
We use $\gamma_0=0.01$ as the default value.

Increasing the temporal stack size generally improved PSNR,
particularly at low particle levels.
Beyond 32 frames, memory use and runtime continued to increase. We therefore use a stack depth of 32 to balance reconstruction
quality and computational cost.

We also compare thinning-based inference modes over multiple aggregation
steps. In the reported Cameraman ablation, the aggregation rule and the
number of evaluations have negligible effect on reconstruction quality.
All main experiments use score-weighted aggregation with \(K=3\) as a fixed
empirical inference setting; direct prediction is included only as the
\(K=1\) ablation.
Extended ablation tables, including exposure conditioning, softmax normalization, and Flower initialization, are reported in the Supplementary Material, Sec.~IV.

\begin{figure}[t]
 \centering
 \includegraphics[width=1.0\linewidth]{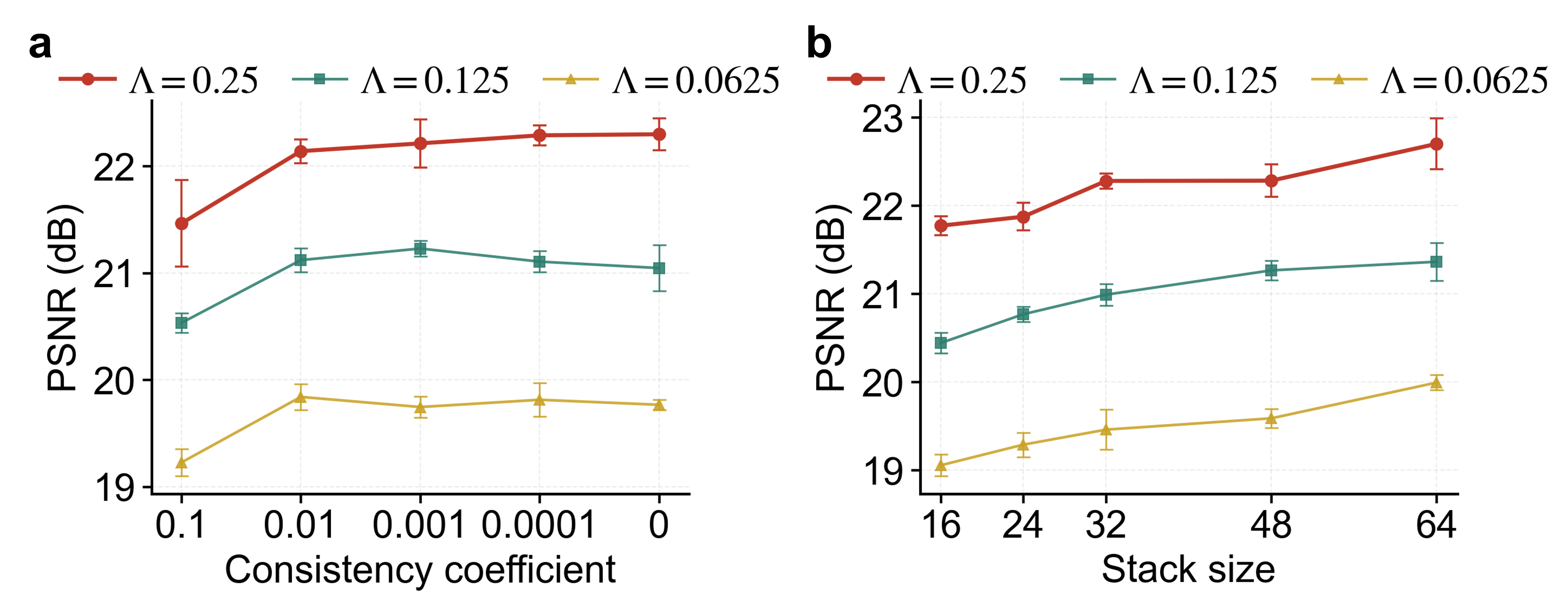}
 \caption{Ablation studies.
 (\textbf{a}) Effect of the base consistency coefficient $\gamma_0$ on the
 3D Cameraman benchmark. A large coefficient degrades the displayed
 particle levels; $\gamma_0=0.01$ is the default used in the main
 experiments.
 (\textbf{b}) Effect of temporal stack size on the 3D Cameraman benchmark.
 Increasing the stack depth generally improves reconstruction quality,
 with larger depths requiring greater memory and runtime cost. Error bars denote the
 standard deviation of PSNR over five independently generated noisy
 realizations.}
 \label{fig:ablation}
\end{figure}

\subsection{Computational Efficiency}
\label{sec:runtime}

We compare wall-clock runtimes on representative video volumes to assess
the computational cost of self-supervised video denoising.
Runtime is reported in minutes.
VBM4D uses its CPU implementation, whereas the learning-based methods run
on an NVIDIA RTX 4090 GPU.
The PARM runtime includes the complete self-supervised denoising procedure
with the 3D U-Net backbone, a batch size of 4, and 10 training epochs.

\begin{table*}[!htbp]
\centering
\caption{
Runtime comparison on different video sequences.
Runtime is reported in minutes.
VBM4D is evaluated on CPU, while N2V, SHINE, UDVD, bit2bit, and PARM are evaluated on GPU.
}
\label{tab:runtime}
\setlength{\tabcolsep}{5pt}
\renewcommand{\arraystretch}{1.08}
\begin{tabular}{||c|c|c||c|c|c|c|c|c||}
\hline
\multicolumn{3}{||c||}{Dataset \textbackslash{} Platform} & CPU & \multicolumn{5}{|c||}{GPU} \\
\hline
Name & Resolution & Frames & VBM4D & N2V & SHINE & UDVD & bit2bit & PARM \\
\hline
Cameraman & \(128 \times 128\) & 513 & 3.03 & 1.00 & 88.14 & 9.72 & 22.52 & 3.30 \\
Coastguard & \(144 \times 176\) & 300 & 2.51 & 2.41 & 53.31 & 13.15 & 41.23 & 10.46 \\
Bus / Miss America & \(288 \times 352\) & 150 & 5.72 & 3.82 & 26.44 & 6.67 & 22.45 & 11.02 \\
Flower / Tennis & \(240 \times 352\) & 150 & 5.28 & 2.31 & 26.38 & 6.68 & 23.25 & 10.48 \\
Foreman & \(288 \times 352\) & 300 & 11.02 & 7.6 & 56.82 & 13.40 & 41.50 & 18.99 \\
\hline
\end{tabular}
\end{table*}

Table~\ref{tab:runtime} shows that PARM falls within the runtime range of
the learning-based self-supervised baselines and is consistently faster
than SHINE.
PARM is slower than N2V and the CPU VBM4D implementation on all reported
sequences, and slower than UDVD on some sequences.
Because VBM4D and the learning-based methods use different hardware, these
measurements characterize wall-clock cost rather than hardware-normalized
algorithmic efficiency.
Together with the reconstruction results above, they quantify the computational
cost associated with PARM's stack-based restoration.

\section{Conclusion}
\label{sec:conclusion}

We introduced PARM, a self-supervised framework for low-count image
denoising based on particle-arrival-rate matching.
Rather than treating denoising as a static regression problem at a
fixed exposure, PARM models particle accumulation as an
exposure-indexed Poisson counting process on a discrete state space.
The Kolmogorov forward equation yields the conditional
particle-arrival-rate field, which equals the posterior mean of the clean
particle-arrival rate and therefore defines the statistical target for
denoising.

For exact calibrated Poisson counts, thinning decomposes a single noisy observation into
complementary lower-exposure views, whose exposure-normalized count
differences provide conditionally unbiased rate targets.
Under the GKL divergence, regression to these self-supervised targets
and supervised clean-rate regression share the same Bayes-optimal
predictor.
In implementation, PARM separates the global count scale from the
normalized spatial allocation, while bidirectional thinning and soft
cross-view consistency improve stability at extremely low counts.

On synthetic Poisson data and standard video benchmarks, PARM achieved
competitive or leading performance across particle levels while retaining fine structures.
On real low-dose electron microscopy data, PARM preserved nanoparticle
structures and showed the strongest agreement with group-averaged
multi-frame references across the evaluated settings.
Because clean microscopy ground truth was unavailable, this comparison
demonstrates reconstruction consistency rather than absolute fidelity
to an unknown clean image.

Future work could extend the framework to mixed Poisson--Gaussian
sensor noise, more efficient multiscale architectures, and real-time
low-count video and microscopy pipelines.
The rate-field perspective links particle-counting models with self-supervised restoration, particularly for microscopy without clean references.

\appendix
\section{Additional Theoretical Details}
\label{app:additional_theory}
\label{app:implementation}

\subsection{Randomized Exposure Without Explicit Conditioning}
\label{app:randomized_exposure}

In the reported implementation, the splitting ratio is randomly sampled
and is not supplied to the network.
Let \(\tau_Z\) denote the random effective exposure associated with a
thinned input, let \(Z\) denote that input, and let \(U\) denote its
ideal complementary rate target.
Conditioned on \(\tau_Z\), Theorem~\ref{thm:finite_difference_equivalence}
gives
\begin{equation}
\mathbb{E}[U\mid Z,\tau_Z]
=
\mathbb{E}[\lambda_0X\mid Z,\tau_Z].
\label{eq:app_conditional_random_exposure}
\end{equation}
Marginalizing over the unobserved exposure state,
\begin{align}
\mathbb{E}[U\mid Z]
&=
\mathbb{E}
\bigl[
\mathbb{E}[U\mid Z,\tau_Z]
\mid Z
\bigr]
\nonumber\\
&=
\mathbb{E}[\lambda_0X\mid Z].
\label{eq:app_marginalized_exposure_equivalence}
\end{align}
Omitting \(\tau_Z\) therefore does not invalidate the population
conditional-mean identity.
Instead, the learned predictor is the Bayes-optimal rate estimator
marginalized over the randomized thinning distribution.
Define
\begin{equation}
v_{\mathrm{mix}}(z)
\equiv
\mathbb{E}[\lambda_0 X \mid Z=z],
\label{eq:app_exposure_marginalized_rate}
\end{equation}
where the expectation is taken under the direction-weighted
randomized-thinning distribution of
Sec.~\ref{sec:self_supervised_thinning}.
This is the Bayes-optimal rate target when the network observes \(Z\)
but not \(\tau_Z\).
The same argument applies to any splitting-ratio distribution
independent of \(X\), including
\(p\sim\mathcal{U}(\varepsilon_p,1/2)\) used in the reported experiments.

The implemented network therefore receives only the noisy count image
or spatiotemporal count stack and outputs an unconstrained score field
\(h_\theta(z)\).
A softmax over \(\Omega\) converts this score field into a normalized
particle-allocation representation, so the network models the relative
spatial allocation of the conditional particle-arrival rate rather than
its global magnitude.

\subsection{Connection to Softmax Allocation}
\label{app:rate_allocation}

Consider a nonnegative rate target \(R\) over the prediction domain
\(\Omega\).
For \(M_R>0\), it can be decomposed as \(R=M_R\pi\) with
\(M_R=\sum_i R_i\), \(\pi_i\geq0\), and \(\sum_i\pi_i=1\).
A positive rate prediction can likewise be written as
\(v=\widetilde m\tilde\pi\) with \(\widetilde m>0\), \(\tilde\pi_i\geq0\),
and \(\sum_i\tilde\pi_i=1\).
Here \(M_R\) is the total mass of the target field \(R\) and is distinct
from both the total conditional rate \(m_\tau\) in
Eq.~\eqref{eq:rate_allocation_decomposition} and the observed particle
mass \(\widehat{M}_\tau\) used at inference.
Substituting these factorizations into the GKL divergence gives
\begin{equation}
D_{\mathrm{GKL}}(R,v)
=
\frac{1}{|\Omega|}
\bigl[
D_{\mathrm{GKL}}^{(1)}(M_R,\widetilde m)
+
M_R D_{\mathrm{KL}}(\pi\|\tilde\pi)
\bigr],
\label{eq:app_gkl_mass_kl}
\end{equation}
where
\(D_{\mathrm{GKL}}^{(1)}(M_R,\widetilde m)=M_R\log(M_R/\widetilde m)-M_R+\widetilde m\)
and
\(D_{\mathrm{KL}}(\pi\|\tilde\pi)=\sum_i\pi_i\log(\pi_i/\tilde\pi_i)\).
When \(M_R=0\), the allocation-dependent term vanishes.

PARM parameterizes the normalized allocation using the unconstrained
network score field \(h_\theta(z)\) followed by a softmax over \(\Omega\),
\(\pi_{\theta,i}(z)=\exp(h_{\theta,i}(z))/\sum_j\exp(h_{\theta,j}(z))\).
For a nonnegative target \(R\), the component of the GKL objective that
depends on the spatial allocation is, up to terms independent of
\(\theta\),
\(\mathcal{L}_{\mathrm{alloc}}=-|\Omega|^{-1}\sum_i R_i\log \pi_{\theta,i}\).
Thus, the target magnitude is retained as a sample-dependent weight,
whereas the network prediction represents only the relative spatial
allocation.

To characterize the population optimum, let \(Z\) denote the network
input and define \(\mu_i(z)=\mathbb{E}[U_i\mid Z=z]\).
The conditional allocation-dependent risk is
\(-\sum_i\mu_i(z)\log\tilde\pi_i\) subject to \(\sum_i\tilde\pi_i=1\).
Whenever \(\sum_j\mu_j(z)>0\), its unique minimizer is
\begin{equation}
\pi_i^{*}(z)
=
\frac{\mathbb{E}[U_i\mid Z=z]}{\sum_{j\in\Omega}\mathbb{E}[U_j\mid Z=z]}.
\label{eq:app_optimal_allocation}
\end{equation}
If \(\sum_j\mu_j(z)=0\), the conditional allocation risk is zero for
every \(\tilde\pi\) on the simplex, so the allocation is not identifiable.

For the ideal finite-exposure target with a fixed thinning ratio
\(p\), Theorem~\ref{thm:finite_difference_equivalence} therefore gives
\begin{equation}
\pi_i^{*}(y)
=
\frac{\mathbb{E}[\lambda_0X_i\mid Y_{p\tau}=y]}
{\sum_{j\in\Omega}\mathbb{E}[\lambda_0X_j\mid Y_{p\tau}=y]}.
\label{eq:app_normalized_posterior_rate}
\end{equation}
Hence, the softmax output identifies the normalized allocation of the
posterior particle-arrival rate, while the global particle mass is
restored separately during inference.

When the thinning ratio is not supplied to the network, the
corresponding optimal softmax output is the normalized spatial
allocation of the exposure-marginalized rate \(v_{\mathrm{mix}}(z)\) in
Eq.~\eqref{eq:app_exposure_marginalized_rate}.

In implementation, the ideal targets \(R^{A}=B/(q\tau)\) and
\(R^{B}=A/(p\tau)\) are replaced by the empirical unit-mean targets
\(\widetilde R^{A}=B/\overline B\) and \(\widetilde R^{B}=A/\overline A\)
defined in Sec.~\ref{sec:practical_allocation}.
Because each empirical mean is a positive scalar shared by all
positions in the corresponding minibatch, this normalization preserves
the within-batch spatial allocation of each complementary count field.
The exact population-equivalence theorem applies to the ideal
exposure-normalized targets; empirical normalization is a practical
implementation-level approximation.

\subsection{Scope and Limitations of the Theoretical Guarantees}
\label{app:theoretical_scope}
\label{app:score_consistency}

The Kolmogorov forward derivation and the finite-exposure argument
identify the same conditional particle-arrival-rate target
\(v_\tau(y)=\mathbb{E}[\lambda_0X\mid Y_\tau=y]\).
The equivalence result is a population statement about the Bayes-optimal
predictor.
It does not assert equality of the numerical risk values, finite-sample
generalization, parameter consistency, or optimization convergence.

The exact result assumes calibrated integer-valued Poisson counts with
known exposure intervals and independent increments, finite first
moments, and unrestricted measurable predictors of finite risk
optimized under GKL divergence.
Randomizing the thinning ratio without providing it to the network
preserves the conditional-mean identity when the ratio is sampled
independently of the latent clean image, but the resulting predictor is
marginalized over the randomized exposure distribution.

The consistency regularizer is not implied by the finite-exposure
equivalence theorem.
The two Bayes predictors \(v_{p\tau}(Y_{p\tau})\) and
\(v_{q\tau}(Y_{q\tau})\) are generally not identical pointwise because
they condition on different observations at different effective
exposures.
Moreover, because the softmax is invariant to a global logit offset, the
raw-score penalty constrains an otherwise nonidentifiable component of
the representation.
The regularizer should therefore be interpreted as a weak empirical
stabilization prior that may bias the population optimum.

Empirical target normalization, randomized rounding of non-integer
measurements, softmax allocation, and observed-mass-preserving
inference are practical implementation components rather than
consequences of the equivalence theorem or forms of posterior
calibration.
Applying the trained network directly to the complete observation
\(Y_\tau\) is likewise a practical inference rule rather than a
consequence of the theorem, whose population analysis concerns inputs
drawn from the randomized thinning distribution.
When read noise, detector gain, imperfect exposure calibration, or
preprocessing violates the calibrated Poisson model, the objective
should be interpreted as a Poisson-informed self-supervised surrogate
rather than an exact unbiased-risk construction.

\ifCLASSOPTIONcaptionsoff
 \newpage
\fi

\footnotesize
\bibliographystyle{IEEEtran}
\bibliography{refs}

\begin{thebibliography}{10}
\providecommand{\url}[1]{#1}
\csname url@samestyle\endcsname
\providecommand{\newblock}{\relax}
\providecommand{\bibinfo}[2]{#2}
\providecommand{\BIBentrySTDinterwordspacing}{\spaceskip=0pt\relax}
\providecommand{\BIBentryALTinterwordstretchfactor}{4}
\providecommand{\BIBentryALTinterwordspacing}{\spaceskip=\fontdimen2\font plus
\BIBentryALTinterwordstretchfactor\fontdimen3\font minus \fontdimen4\font\relax}
\providecommand{\BIBforeignlanguage}[2]{{%
\expandafter\ifx\csname l@#1\endcsname\relax
\typeout{** WARNING: IEEEtran.bst: No hyphenation pattern has been}%
\typeout{** loaded for the language `#1'. Using the pattern for}%
\typeout{** the default language instead.}%
\else
\language=\csname l@#1\endcsname
\fi
#2}}
\providecommand{\BIBdecl}{\relax}
\BIBdecl

\bibitem{shangguan2023day}
M.~Shangguan, Z.~Weng, Z.~Lin, Z.~Lee, M.~Shangguan, Z.~Yang, J.~Sun, T.~Wu, Y.~Zhang, and C.~Wen, ``Day and night continuous high-resolution shallow-water depth detection with single-photon underwater lidar,'' \emph{Opt. Express}, vol.~31, no.~26, pp. 43\,950--43\,962, 2023.

\bibitem{rapp2020advances}
J.~Rapp, J.~Tachella, Y.~Altmann, S.~McLaughlin, and V.~K. Goyal, ``Advances in single-photon lidar for autonomous vehicles: Working principles, challenges, and recent advances,'' \emph{IEEE Signal Process. Mag.}, vol.~37, no.~4, pp. 62--71, 2020.

\bibitem{morris2015imaging}
P.~A. Morris, R.~S. Aspden, J.~E. Bell, R.~W. Boyd, and M.~J. Padgett, ``Imaging with a small number of photons,'' \emph{Nat. Commun.}, vol.~6, no.~1, p. 5913, 2015.

\bibitem{em}
P.~A. Crozier, M.~Leibovich, P.~Haluai, M.~Tan, A.~M. Thomas, J.~Vincent, S.~Mohan, A.~Marcos~Morales, S.~A. Kulkarni, D.~S. Matteson \emph{et~al.}, ``Visualizing nanoparticle surface dynamics and instabilities enabled by deep denoising,'' \emph{Science}, vol. 387, no. 6737, pp. 949--954, 2025.

\bibitem{shine}
J.~Kim, J.~Rhee, S.~Kang, M.~Jung, J.~Kim, M.~Jeon, J.~Park, J.~Ham, B.~H. Kim, W.~C. Lee \emph{et~al.}, ``Self-supervised machine learning framework for high-throughput electron microscopy,'' \emph{Sci. Adv.}, vol.~11, no.~14, p. eads5552, 2025.

\bibitem{le2007variational}
T.~Le, R.~Chartrand, and T.~J. Asaki, ``A variational approach to reconstructing images corrupted by poisson noise,'' \emph{J. Math. Imag. Vis.}, vol.~27, no.~3, pp. 257--263, 2007.

\bibitem{allard2008total}
W.~K. Allard, ``Total variation regularization for image denoising, i. geometric theory,'' \emph{SIAM J. Math. Anal.}, vol.~39, no.~4, pp. 1150--1190, 2008.

\bibitem{anscombe1948transformation}
F.~J. Anscombe, ``The transformation of poisson, binomial and negative-binomial data,'' \emph{Biometrika}, vol.~35, no. 3/4, pp. 246--254, 1948.

\bibitem{purelet}
F.~Luisier, T.~Blu, and M.~Unser, ``Image denoising in mixed poisson--gaussian noise,'' \emph{IEEE Trans. Image Process.}, vol.~20, no.~3, pp. 696--708, 2011.

\bibitem{chen2018learning}
C.~Chen, Q.~Chen, J.~Xu, and V.~Koltun, ``Learning to see in the dark,'' in \emph{Proc. IEEE/CVF Conf. Comput. Vis. Pattern Recognit. (CVPR)}, 2018, pp. 3291--3300.

\bibitem{li2025ts}
Y.~Li, Z.~Zhang, J.~Xia, J.~Cheng, Q.~Wu, J.~Li, Y.~Tian, and H.~Kong, ``Ts-diff: Two-stage diffusion model for low-light raw image enhancement,'' in \emph{Proc. Int. Joint Conf. Neural Netw. (IJCNN)}, 2025, pp. 1--8.

\bibitem{lehtinen2018noise2noise}
J.~Lehtinen, J.~Munkberg, J.~Hasselgren, S.~Laine, T.~Karras, M.~Aittala, and T.~Aila, ``{N}oise2{N}oise: Learning image restoration without clean data,'' in \emph{Proc. Int. Conf. Mach. Learn. (ICML)}, ser. Proc. Mach. Learn. Res., J.~Dy and A.~Krause, Eds., vol.~80.\hskip 1em plus 0.5em minus 0.4em\relax PMLR, 10--15 Jul 2018, pp. 2965--2974.

\bibitem{krull2019noise2void}
A.~Krull, T.-O. Buchholz, and F.~Jug, ``Noise2void-learning denoising from single noisy images,'' in \emph{Proc. IEEE/CVF Conf. Comput. Vis. Pattern Recognit. (CVPR)}, 2019, pp. 2129--2137.

\bibitem{batson2019noise2self}
J.~Batson and L.~Royer, ``Noise2self: Blind denoising by self-supervision,'' in \emph{Proc. Int. Conf. Mach. Learn. (ICML)}.\hskip 1em plus 0.5em minus 0.4em\relax PMLR, 2019, pp. 524--533.

\bibitem{stbn}
Z.~Chen, T.~Jiang, X.~Hu, W.~Zhang, H.~Li, and H.~Wang, ``Spatiotemporal blind-spot network with calibrated flow alignment for self-supervised video denoising,'' in \emph{Proc. AAAI Conf. Artif. Intell. (AAAI)}, vol.~39, no.~3, 2025, pp. 2411--2419.

\bibitem{lipman2023flow}
Y.~Lipman, R.~T.~Q. Chen, H.~Ben-Hamu, M.~Nickel, and M.~Le, ``Flow matching for generative modeling,'' in \emph{Proc. Int. Conf. Learn. Represent. (ICLR)}, 2023.

\bibitem{holderrieth2025generator}
P.~Holderrieth, M.~Havasi, J.~Yim, N.~Shaul, I.~Gat, T.~Jaakkola, B.~Karrer, R.~T.~Q. Chen, and Y.~Lipman, ``Generator matching: Generative modeling with arbitrary markov processes,'' in \emph{Proc. Int. Conf. Learn. Represent. (ICLR)}, 2025.

\bibitem{bartlett1978introduction}
M.~S. Bartlett, \emph{An introduction to stochastic processes: with special reference to methods and applications}.\hskip 1em plus 0.5em minus 0.4em\relax CUP Archive, 1978.

\bibitem{rudin1992nonlinear}
L.~I. Rudin, S.~Osher, and E.~Fatemi, ``Nonlinear total variation based noise removal algorithms,'' \emph{Phys. D: Nonlinear Phenom.}, vol.~60, no. 1-4, pp. 259--268, 1992.

\bibitem{luisier2010fast}
F.~Luisier, C.~Vonesch, T.~Blu, and M.~Unser, ``Fast interscale wavelet denoising of poisson-corrupted images,'' \emph{Signal Process.}, vol.~90, no.~2, pp. 415--427, 2010.

\bibitem{blu2007sure}
T.~Blu and F.~Luisier, ``The {SURE-LET} approach to image denoising,'' \emph{IEEE Trans. Image Process.}, vol.~16, no.~11, pp. 2778--2786, 2007.

\bibitem{dabov2007image}
K.~Dabov, A.~Foi, V.~Katkovnik, and K.~Egiazarian, ``Image denoising by sparse 3-d transform-domain collaborative filtering,'' \emph{IEEE Trans. Image Process.}, vol.~16, no.~8, pp. 2080--2095, 2007.

\bibitem{makinen2020collaborative}
Y.~M{\"a}kinen, L.~Azzari, and A.~Foi, ``Collaborative filtering of correlated noise: Exact transform-domain variance for improved shrinkage and patch matching,'' \emph{IEEE Trans. Image Process.}, vol.~29, pp. 8339--8354, 2020.

\bibitem{maggioni2012nonlocal}
M.~Maggioni, V.~Katkovnik, K.~Egiazarian, and A.~Foi, ``Nonlocal transform-domain filter for volumetric data denoising and reconstruction,'' \emph{IEEE Trans. Image Process.}, vol.~22, no.~1, pp. 119--133, 2012.

\bibitem{makitalo2010optimal}
M.~Makitalo and A.~Foi, ``Optimal inversion of the anscombe transformation in low-count poisson image denoising,'' \emph{IEEE Trans. Image Process.}, vol.~20, no.~1, pp. 99--109, 2010.

\bibitem{makitalo2013exact}
M.~Mäkitalo and A.~Foi, ``Poisson-gaussian denoising using the exact unbiased inverse of the generalized anscombe transformation,'' in \emph{Proc. IEEE Int. Conf. Acoust., Speech Signal Process. (ICASSP)}, 2012, pp. 1081--1084.

\bibitem{shen2018improved}
Y.~Shen, Y.~Chen, Q.~Liu, S.~Lou, W.~Yu, X.~Wang, and H.~Chen, ``Improved anscombe transformation and total variation for denoising of lowlight infrared images,'' \emph{Infrared Phys. Technol.}, vol.~93, pp. 192--198, 2018.

\bibitem{zhang2021poisson}
J.~Zhang, L.~Zhao, and Z.~Wei, ``Poisson-skellam distribution based regularization conditional random field method for photon-limited poisson image denoising,'' \emph{Signal Process.}, vol. 188, p. 108165, 2021.

\bibitem{foi2008practical}
A.~Foi, M.~Trimeche, V.~Katkovnik, and K.~Egiazarian, ``Practical poissonian-gaussian noise modeling and fitting for single-image raw-data,'' \emph{IEEE Trans. Image Process.}, vol.~17, no.~10, pp. 1737--1754, 2008.

\bibitem{zhang2019poisson}
Y.~Zhang, Y.~Zhu, E.~Nichols, Q.~Wang, S.~Zhang, C.~Smith, and S.~Howard, ``A poisson-gaussian denoising dataset with real fluorescence microscopy images,'' in \emph{Proc. IEEE/CVF Conf. Comput. Vis. Pattern Recognit. (CVPR)}, 2019, pp. 11\,710--11\,718.

\bibitem{ramani2008monte}
S.~Ramani, T.~Blu, and M.~Unser, ``Monte-carlo {SURE}: A black-box optimization of regularization parameters for general denoising algorithms,'' \emph{IEEE Trans. Image Process.}, vol.~17, no.~9, pp. 1540--1554, 2008.

\bibitem{lemontagner2014unbiased}
Y.~Le~Montagner, E.~D. Angelini, and J.-C. Olivo-Marin, ``An unbiased risk estimator for image denoising in the presence of mixed poisson--gaussian noise,'' \emph{IEEE Trans. Image Process.}, vol.~23, no.~3, pp. 1255--1268, 2014.

\bibitem{krull2020probabilistic}
A.~Krull, T.~Vi{\v{c}}ar, M.~Prakash, M.~Lalit, and F.~Jug, ``Probabilistic noise2void: Unsupervised content-aware denoising,'' \emph{Front. Comput. Sci.}, vol.~2, p.~5, 2020.

\bibitem{xie2020noise2same}
Y.~Xie, Z.~Wang, and S.~Ji, ``Noise2same: Optimizing a self-supervised bound for image denoising,'' \emph{Adv. Neural Inf. Process. Syst. (NeurIPS)}, vol.~33, pp. 20\,320--20\,330, 2020.

\bibitem{quan2020self2self}
Y.~Quan, M.~Chen, T.~Pang, and H.~Ji, ``{Self2Self} with dropout: Learning self-supervised denoising from single image,'' in \emph{Proc. IEEE/CVF Conf. Comput. Vis. Pattern Recognit. (CVPR)}, 2020, pp. 1887--1895.

\bibitem{wang2022blind2unblind}
Z.~Wang, J.~Liu, G.~Li, and H.~Han, ``Blind2unblind: Self-supervised image denoising with visible blind spots,'' in \emph{Proc. IEEE/CVF Conf. Comput. Vis. Pattern Recognit. (CVPR)}, 2022, pp. 2027--2036.

\bibitem{huang2021neighbor2neighbor}
T.~Huang, S.~Li, X.~Jia, H.~Lu, and J.~Liu, ``Neighbor2neighbor: Self-supervised denoising from single noisy images,'' in \emph{Proc. IEEE/CVF Conf. Comput. Vis. Pattern Recognit. (CVPR)}, 2021, pp. 14\,781--14\,790.

\bibitem{lequyer2022fast}
J.~Lequyer, R.~Philip, A.~Sharma, W.-H. Hsu, and L.~Pelletier, ``A fast blind zero-shot denoiser,'' \emph{Nat. Mach. Intell.}, vol.~4, no.~11, pp. 953--963, 2022.

\bibitem{mansour2023zero}
Y.~Mansour and R.~Heckel, ``Zero-shot {Noise2Noise}: Efficient image denoising without any data,'' in \emph{Proc. IEEE/CVF Conf. Comput. Vis. Pattern Recognit. (CVPR)}, 2023, pp. 14\,018--14\,027.

\bibitem{ma2025pixel2pixel}
Q.~Ma, J.~Jiang, X.~Zhou, P.~Liang, X.~Liu, and J.~Ma, ``Pixel2pixel: A pixelwise approach for zero-shot single image denoising,'' \emph{IEEE Trans. Pattern Anal. Mach. Intell.}, vol.~47, no.~6, pp. 4614--4629, 2025.

\bibitem{udvd}
D.~Y. Sheth, S.~Mohan, J.~L. Vincent, R.~Manzorro, P.~A. Crozier, M.~M. Khapra, E.~P. Simoncelli, and C.~Fernandez-Granda, ``Unsupervised deep video denoising,'' in \emph{Proc. IEEE/CVF Int. Conf. Comput. Vis. (ICCV)}, 2021, pp. 1759--1768.

\bibitem{aiyetigbo2024umvd}
M.~Aiyetigbo, A.~Korte, E.~Anderson, R.~Chalhoub, P.~Kalivas, F.~Luo, and N.~Li, ``Unsupervised microscopy video denoising,'' in \emph{Proc. IEEE/CVF Conf. Comput. Vis. Pattern Recognit. Workshops (CVPRW)}, 2024.

\bibitem{chen2026ssd}
Y.~Chen, S.~Wang, Y.~Wang, J.~Mao, L.~Liu, X.~Cao, Z.~Chen \emph{et~al.}, ``Self-optimized spectral distance for low-light high-throughput raman hyperspectral imaging,'' \emph{Nat. Comput. Sci.}, vol.~6, pp. 603--616, 2026.

\bibitem{bit2bit}
Y.~Liu, A.~Krull, H.~Basevi, A.~Leonardis, and M.~W. Jenkins, ``bit2bit: 1-bit quanta video reconstruction via self-supervised photon prediction,'' in \emph{Adv. Neural Inf. Process. Syst. (NeurIPS)}, 2024.

\bibitem{khademi2021self}
W.~Khademi, S.~Rao, C.~Minnerath, G.~Hagen, and J.~Ventura, ``Self-supervised poisson-gaussian denoising,'' in \emph{Proc. IEEE/CVF Winter Conf. Appl. Comput. Vis. (WACV)}, 2021, pp. 2131--2139.

\bibitem{ta2022poisson2sparse}
C.-K. Ta, A.~Aich, A.~Gupta, and A.~K. Roy-Chowdhury, ``Poisson2sparse: Self-supervised poisson denoising from a single image,'' in \emph{Proc. Int. Conf. Med. Image Comput. Comput.-Assist. Intervent. (MICCAI)}.\hskip 1em plus 0.5em minus 0.4em\relax Springer, 2022, pp. 557--567.

\bibitem{monroy2025generalized}
B.~Monroy, J.~Bacca, and J.~Tachella, ``Generalized recorrupted-to-recorrupted: Self-supervised learning beyond gaussian noise,'' in \emph{Proc. IEEE/CVF Conf. Comput. Vis. Pattern Recognit. (CVPR)}, 2025, pp. 28\,155--28\,164.

\bibitem{albergo2025stochastic}
M.~Albergo, N.~M. Boffi, and E.~Vanden-Eijnden, ``Stochastic interpolants: A unifying framework for flows and diffusions,'' \emph{J. Mach. Learn. Res.}, vol.~26, no. 209, pp. 1--80, 2025.

\bibitem{pourya2026flower}
M.~Pourya, B.~El~Rawas, and M.~Unser, ``{Flower}: A flow-matching solver for inverse problems,'' in \emph{Proc. Int. Conf. Learn. Represent. (ICLR)}, 2026.

\bibitem{gat2024discrete}
I.~Gat, T.~Remez, N.~Shaul, F.~Kreuk, R.~T. Chen, G.~Synnaeve, Y.~Adi, and Y.~Lipman, ``Discrete flow matching,'' \emph{Adv. Neural Inf. Process. Syst. (NeurIPS)}, vol.~37, pp. 133\,345--133\,385, 2024.

\bibitem{chen2023learning}
T.~Chen and M.~Zhou, ``Learning to jump: Thinning and thickening latent counts for generative modeling,'' in \emph{Proc. Int. Conf. Mach. Learn. (ICML)}, 2023.

\bibitem{bhattacharya2025itdpdm}
S.~Bhattacharya, A.~Gorle, A.~Bilal, C.~Ding, A.~K.~S. Yadav, and T.~Weissman, ``{ItDPDM}: Information-theoretic discrete poisson diffusion model,'' in \emph{Adv. Neural Inf. Process. Syst. (NeurIPS)}, 2025.

\bibitem{maggioni2012vbm4d}
M.~Maggioni, G.~Boracchi, A.~Foi, and K.~Egiazarian, ``Video denoising, deblocking, and enhancement through separable 4-d nonlocal spatiotemporal transforms,'' \emph{IEEE Trans. Image Process.}, vol.~21, no.~9, pp. 3952--3966, 2012.

\end{thebibliography}

\end{document}